\documentclass[11pt]{article} 
	\usepackage[margin=1.0 in]{geometry} 
	\usepackage[english]{babel} 
	\usepackage{authblk}
	\usepackage{amsmath, amssymb, textcomp, amsthm, mathtools,mathrsfs} 
	\allowdisplaybreaks
	\usepackage{wasysym} 
	\usepackage{stmaryrd}
	\usepackage{dsfont}
	\usepackage{hyperref}
   	\usepackage{natbib}
   	\usepackage{enumitem}
   	\usepackage{inputenc}
	\usepackage{csquotes}
	\usepackage{footnote,footmisc} 
	\usepackage{setspace} 
	\usepackage{threeparttable}
	\usepackage{tikz}
	\usepackage{graphicx}
	\usepackage{pgfplots}
	\pgfplotsset{compat=1.17}
	\usetikzlibrary{shapes.misc}
	\usetikzlibrary{patterns}
	\usepackage{xcolor}
	\hypersetup{
    	colorlinks,
    	linkcolor={red!50!black},
    	citecolor={blue!50!black},
    	urlcolor={blue!80!black}
	}
	\usetikzlibrary{decorations.pathreplacing}
	\usetikzlibrary{arrows}
	\usepackage{booktabs} 
	\usepackage{multirow}
	\usepackage{float} 
	\newcommand*\dd{\mathop{}\!\mathrm{d}}
    
	\usepackage[toc,page]{appendix}
	
	\newtheorem{theorem}{Theorem}
	\newtheorem{lemma}{Lemma} 
	\newtheorem{proposition}{Proposition} 
	\newtheorem{definition}{Definition}

	\newtheorem{corollary}{Corollary}
	\newtheorem{assumption}{Assumption}

	\newtheorem*{example*}{Example}
	\usepackage{graphicx}  
	\usepackage{marvosym} 
	\newcommand\independent{\protect\mathpalette{\protect\independenT}{\perp}}
    \def\independenT#1#2{\mathrel{\rlap{$#1#2$}\mkern2mu{#1#2}}}
    \DeclareMathOperator*{\argmin}{\arg\!\min}

    \newtheorem*{assumptions*}{\assumptionnumber}
\providecommand{\assumptionnumber}{}

\usepackage{caption}
\usepackage{subcaption}
\usepackage{mwe}

\definecolor{cyan}{cmyk}{1, 0.4, 0, 0}

\definecolor{mypink}{RGB}{219, 48, 122}

\makeatletter

 \def\thanks#1{\protected@xdef\@thanks{\@thanks
        \protect\footnotetext{#1}}}
\makeatother
 
\title{Matching for causal effects via multimarginal unbalanced optimal transport}
\author[1]{Florian Gunsilius\thanks{Authors listed alphabetically. FG is supported by a MITRE faculty research award. The authors thank Donald W.~Andrews, L\'ena\"ic Chizat, Antonio Galvao, Augusto Gerolin, Meng Hsuan Hsieh, Jun Kitagawa, Marcel Nutz, and Jeffrey Wooldridge, and participants at the IMSI conference on Applied Optimal Transport, the Rochester Econometrics Mini-Conference, LSE, Michigan State University, the University of Glasgow, the University of Ottawa, and Yale University for helpful comments. A previous version was titled ``Matching for causal effects via multimarginal optimal transport''. All errors are the authors'.}}
\author[2]{Yuliang Xu}
\affil[1]{Department of Economics, University of Michigan}
\affil[2]{Department of Biostatistics, University of Michigan}
\affil[ ]{\texttt{\{ffg, yuliangx\}@umich.edu}}
\date{\today}
\begin{document}
\maketitle
\begin{abstract}
Matching on covariates is a well-established framework for estimating causal effects in observational studies. The principal challenge stems from the often high-dimensional structure of the problem. Many methods have been introduced to address this, with different advantages and drawbacks in computational and statistical performance as well as interpretability. This article introduces a natural optimal matching method based on multimarginal unbalanced optimal transport that possesses many useful properties in this regard. It provides interpretable weights based on the distance of matched individuals, can be efficiently implemented via the iterative proportional fitting procedure, and can match several treatment arms simultaneously. Importantly, the proposed method only selects good matches from either group, hence is competitive with the classical $k$-nearest neighbors approach in terms of bias and variance in finite samples. Moreover, we prove a central limit theorem for the empirical process of the potential functions of the optimal coupling in the unbalanced optimal transport problem with a fixed penalty term. This implies a parametric rate of convergence of the empirically obtained weights to the optimal weights in the population for a fixed penalty term.
\end{abstract}
\section{Introduction}
Identifying cause and effect is one of the primary goals of scientific research. The leading approaches to uncover causal effects are randomized controlled trials. Unfortunately, such trials are often practically infeasible on ethical grounds, might not be generalizable beyond the experimental setting due to lack of variation in the population, or simply have too few participants to generate robust results due to financial or logistical restrictions. An attractive alternative is to use observational data, which are ubiquitous, often readily available, and comprehensive.

The main challenge in using observational data for causal inference is the fact that assignment into treatment is not perfectly randomized. This implies that individuals assigned to different treatments may possess systematically different observable and unobservable covariates. Comparing the outcomes between individuals in different treatment groups may then yield a systematically biased estimator of the true causal effect. Matching methods are designed to balance the treatment samples in such a way that differences between the observed covariates of the groups are minimized. This allows the researcher to directly compare the balanced treatment groups for estimating the true causal effect under the assumption that the unobservable covariates of individuals are similar if their observed covariates are similar.
\subsection{Existing approaches and persisting challenges}
The most widely adopted matching approach in the setting of a binary treatment is matching on the propensity score \citep{rosenbaum1983central}. Many adjustments and refinements have been proposed in the literature, see for instance \citet{austin2008critical, austin2011introduction} and \citet{imbens2004nonparametric} for reviews. The idea is to balance the samples of treatment and control groups by comparing individuals with similar propensity scores, i.e.~the probability of receiving the treatment conditional on the observable covariates. Under the assumption that the observed covariates capture all the systematic heterogeneity of individuals, the so-called ``unconfoundedness assumption'' \citep{imbens2000role, imbens2015causal}, the researcher can obtain an estimate of the respective causal effect by comparing the outcomes of individuals in treatment and control group with similar propensity scores. The utility of the propensity score derives from dimension reduction: instead of having to match individual covariates in a possibly high-dimensional space, one can match on a one-dimensional score. 

However, propensity score matching is no panacea for unbalanced samples. First, as pointed out in \citet{imbens2000role}, in multivalued treatment settings comparing individuals across groups with similar propensity scores is not equivalent to comparing individuals with similar covariates. This induces bias in the procedure and makes the propensity score ``less well-suited'' in a setting with more than two treatment arms \citep{imbens2000role, lopez2017estimation}. Recent approaches address this issue at the cost of interpretability \citep{li2019propensity}. Second, even in binary treatment settings, propensity score matching has been shown to introduce bias in finite samples \citep{king2019propensity}. The reason lies in the fact that propensity score matching tries to generate a balanced sample that resembles a randomized trial. This does not imply that all corresponding covariates are matched exactly in finite samples, see \citet{king2019propensity} for a detailed analysis. Third, in high-dimensional settings, one usually resorts to estimating the propensity score function based on parametric models, which introduces misspecification bias.

The main alternative to propensity score matching is to match the covariates directly. There are a multitude of matching approaches and refinements, see for instance \citet{abadie2004implementing, abadie2011bias, hainmueller2012entropy, lopez2017estimation, morucci2020adaptive, rosenbaum1989optimal, rosenbaum2020modern, rosenbaum2007minimum, zubizarreta2015stable} and references therein. \citet{stuart2010matching} provides a concise overview of the fundamental ideas. \citet{kallus2020generalized} recently introduced a general framework for optimal matching that subsumes many approaches. The most widely used idea is to specify a distance, e.g.~a weighted Euclidean distance or some discrepancy measure like the Kullback-Leibler divergence, and proceed via some form of $k$-nearest neighbor matching \citep{abadie2004implementing, rubin1973matching, stuart2010matching}. 

Direct matching comes with computational and statistical challenges. From a computational perspective, matching covariates in high-dimensional settings is more complex than matching on a one-dimensional score, see for instance \citet{rosenbaum1989optimal} for a classical analysis. From a statistical perspective, most optimal matching approaches, and especially $k$-nearest neighbors matching approaches or variants thereof, are non-smooth approaches due to the strict cut-off at $k$ neighbors. Therefore, they often provide optimal matchings with poor statistical properties. In particular, the obtained estimators for causal effects are often asymptotically biased and need not converge at the parametric rate as shown in \citet{abadie2006large} in the case of $k$-nearest neighbor approaches. \citet{abadie2011bias} introduce a correction for these estimators based on linear regression. Analogous rigorous statistical results for similar matching algorithms have been absent in the literature, especially in the multivalued treatment setting as argued in \citet{lopez2017estimation}. 

There is hence a need for a direct matching method that can provide interpretable weights of the optimal matches across several treatment arms while possessing good computational and statistical properties.

\subsection{Outline of the proposed method}

This article introduces a natural matching method based on unbalanced optimal transport. It is an alternative to existing matching approaches in terms of how it weights the matches. It is similar in spirit to the nearest-neighbor matching idea, with the important difference that the number of good matches and corresponding weights are chosen endogenously. The idea of optimal transportation is to construct a coupling $\gamma$ between the distributions of the covariate distributions of the respective treatment and control groups. This coupling introduces a natural weighting scheme by considering the respective conditional measures $\gamma(A|x_1)$ for each treatment group, where $A$ is the set of control individuals for the given treated individual with covariate $X_1=x_1$. 

This weighting is what sets the proposed method apart from other matching approaches, especially the classical $k$-nearest neighbor approach, where one chooses a fixed number $k$ of neighbors and weights them uniformly, a non-smooth approach for given data \citep{abadie2006large}. By contrast, the proposed method chooses a natural weighting $\gamma(\cdot |x_1)$ of the matched control individuals that takes into account the cost of matching. One commonly chooses distances on the covariate space as the cost function, for instance the Euclidean distance or weighted counterparts. Then the proposed method naturally re-weights the matched control units in terms of their distance from the treated unit. This provides a natural and smooth weighting procedure and is the main reason why our method enjoys good statistical performance, as we prove in section \ref{sec:stats} below. 

Moreover, and equally importantly, the proposed method is based on unbalanced optimal transport \citep{chizat2019unbalanced, chizat2018scaling, sejourne2019sinkhorn, liero2018optimal}, which enables it to endogenously choose optimal matchings and drop observations in either treatment- or control group that do not have a good match in the other group. This is different from the $k$-nearest neighbor approach, where one matches every observation in one group and only drops observations from the other group that are bad matches. The fact that our proposed method can drop bad potential matches from both groups makes it competitive in terms of bias in finite samples with respect to $k$-nearest neighbor approaches. 

The concepts of unbalanced and partial optimal transport were introduced to allow for partial matches without the constraint from classical optimal transport that requires researchers to match every individual in the dataset \citep{caffarelli2010free, figalli2010optimal, kitagawa2015multi}. The utility of the unbalanced optimal transport approach stems from the fact that it relies on entropy regularization to achieve this and can hence be implemented efficiently in practice via versions of the standard \emph{Iterative Proportion Fitting Procedure} (IPFP) \citep{deming1940least}, also called the Sinkhorn algorithm \citep{sinkhorn1964relationship}. This procedure can straightforwardly be extended to the multiple-treatment setting, which we analyze below. In this sense, the proposed method addresses the challenges outlined in the introduction: it (i) provides optimal matchings in finite samples, (ii) gives interpretable weights based on covariate distance, (iii) is computationally efficient with (iv) good statistical properties, and (v) can even provide a joint optimal matching in multivalued treatment settings. We demonstrate the excellent finite sample properties in a simple simulation exercise in section \ref{sec:illu}.

We provide the statistical properties of the optimal matching obtained via the proposed procedure in section \ref{sec:stats}. These results are valid when there are different numbers of individuals in different treatment arms, which is crucial in practical settings. From a mathematical perspective, it complements the analytical and statistical results in \citet{sejourne2019sinkhorn} by showing that the optimal weights of the matches obtained via this procedure converge at the parametric rate for a fixed penalty term. It is built on a central limit theorem for the potentials of the dual problem of the unbalanced optimal transport problem. This result also complements other existing statistical results for the Sinkhorn alogorithm for classical entropy-regularized balanced optimal transport problems from \citet{harchaoui2020asymptotics}, who require that the number of individuals in both treatment arms are the same. However, they derive an expression for the asymptotic variance of the estimator via Gaussian chaos decompositions, which we do not do.

In concurrent work to ours, \citet{dunipace2021optimal} provides a matching estimator that is also based on optimal transportation. Their setting is a prediction problem, which is different from ours. They use optimal transport matching to impute the outcomes and the corresponding effects. In contrast, we use optimal transportation to jointly balance and re-weight several samples that only receive one type of treatment. In particular, a main result of this paper is to introduce a balancing of the groups via the conditional measure $\gamma(\cdot|x)$ obtained by the joint coupling $\gamma$. Our proposed method hence redefines the estimation for causal effects in a way that is different from existing approaches.  

\section{Matching via unbalanced optimal transport}
This section briefly introduces the mathematical background on optimal transport needed for the proposed matching estimator. For further information we refer to \citet{santambrogio2015optimal} and \citet{villani2003topics,villani2008optimal} for introductions to general optimal transportation theory and to \citet{peyre2019computational} for a review of the divergence-regularized optimal transportation and unbalanced optimal transport approach we consider.

\subsection{Optimal transport: a brief overview}
For two given marginal distributions $\mu_0$ and $\mu_1$ of characteristics of groups to be matched, defined on spaces $\mathcal{X}$ and $\mathcal{Y}$ respectively, the classical Kantorovich problem \citep[chapter 1]{villani2003topics} consists in finding an optimal transport plan, i.e.~a probabilistic matching between the group whose characteristics are distributed as $\mu_0$ and the group whose characteristics are distributed as $\mu_1$. This is a joint coupling $\gamma$ in the space of all couplings $\Pi(\mu_0,\mu_1)$ of these two marginals in such a way that the average cost of transportation of mass from $\mathcal{X}$ to $\mathcal{Y}$ is minimized. Formally, the Kantorovich problem takes the form
\begin{equation}\label{eq:kantorovich}
\inf_{\gamma\in\Pi(\mu,\nu)} \int_{\mathcal{X}\times\mathcal{Y}} c(x,y) \dd\gamma(x,y),
\end{equation} 
where $\Pi(\mu,\nu)$ is the set of all (Borel-) probability measures on $\mathcal{X}\times\mathcal{Y}$ whose marginals are $\mu$ and $\nu$, respectively, and where $c:\mathcal{X}\times\mathcal{Y}\to\mathbb{R}$ is a cost function encoding the difficulty of moving mass from $\mathcal{X}$ to $\mathcal{Y}$. In many cases of interest the cost function is chosen to be a metric on the underlying space, such as the Euclidean distance, i.e.
\[c(x,y)\equiv |x-y|^2.\] A distance function such as the squared Euclidean distance is the most natural choice for a cost function, but our method applies to any cost function that is infinitely often differentiable. This can be useful in settings where the researchers want to take into account more information on the cost of comparing individuals than only the distance of their covariates. The Kantorovich problem is a convex relaxation of the Monge problem, which requires that the optimal transportation is achieved via a deterministic mapping $T$, that is, it takes the form
\begin{equation}\label{Eq:Monge_problem}\inf_{T_\#\mu = \nu} \int_{\mathcal{X}\times\mathcal{Y}} c(x,T(x)) \dd\mu(x),
\end{equation} 
where $T_\# P$ denotes the pushforward measure of $P$ via $T:\mathcal{X}\to\mathcal{Y}$. 

This binary setting can straightforwardly be extended to the multimarginal setting by searching for a joint coupling over several marginals $\{\mu_j\}_{j=1,\ldots, J}$ \citep{gangbo1998optimal}:
\begin{equation}\label{eq:classic}
\inf_{\gamma\in\Pi(\mu_1,\ldots,\mu_J)} \int_{\mathcal{X}} c(x) \dd\gamma(x),
\end{equation} where we denote $\mathcal{X} = \mathcal{X}_1,\ldots,\mathcal{X}_J$ and $x\coloneqq (x_1,\ldots, x_J)$. Note that in the multimarginal setting we index the measures $\mu_j$ starting from $1$, in contrast to the binary setting. This is to signify that in a multimarginal setting one does not necessarily have a true control group.

In many cases in the binary setting, for instance when $\mu_0$ possesses a density with respect to Lebesgue measure and the cost function is strictly convex, $\gamma$ is supported on the graph of a function. In this case one obtains a deterministic matching between individuals in $\mu_0$ and individuals in $\mu_1$ that takes the form of a function, i.e.~every individual with characteristic $x\in\mathcal{X}$ is mapped to at most one individual with characteristic $y\in\mathcal{Y}$. A special case of this is Brenier's classical theorem on the structure of optimal solutions of the Monge-Kantorovich problem in the case of the squared Euclidean distance as a cost function \citep{brenier1991polar}. Similar results have been shown in the multimarginal setting, see \citet{agueh2011barycenters, carlier2010matching, gangbo1998optimal, pass2015multi}.

\subsection{Unbalanced optimal transport and divergence regularization}
The classical optimal transport problem is not well-suited for the purpose of matching for causal effects for two reasons. First, its computational and statistical properties are poor. In particular, in practical settings, when the actual distributions $\mu_0$ and $\mu_1$ are unknown and have to be estimated from the observable characteristics of individuals, estimating the optimal matching $\gamma$ suffers from the curse of dimensionality \citep{forrow2019statistical, gunsilius2021convergence, hutter2021minimax, manole2021plugin, weed2019sharp}. This curse of dimensionality is also computational \citep{altschuler2021hardness}. Second, classical optimal transport imposes the strong constraint that all elements in the sample have to be matched. From a causal inference perspective, this induces bias because not all individuals will have a close enough match, and matching individuals that are too different induces bias.

This is why we rely on unbalanced optimal transport \citep{chizat2018scaling, chizat2019unbalanced, sejourne2019sinkhorn, liero2018optimal}, which solves the partial optimal transport problem via divergence regularization. A special case is entropy-regularization of the classical balanced optimal transport problem, which was introduced in \citet{galichon2010matching} in economics and in \citet{cuturi2013sinkhorn} in the machine learning literature and is one of the main ways to compute approximations to optimal transportation matchings in practice \citep{peyre2019computational}.

To obtain a multimarginal matching between the distributions $\mu_1,\ldots,\mu_J$ of characteristics of individuals, we solve the following divergence-regularized optimal transport problem for some $\varepsilon> 0$ and a cost function $c(x)$ which we will assume throughout is smooth, i.e.~$c(x)\equiv c(x_1,\ldots,x_J)\in C^\infty\left(\mathcal{X}\right)$ \citep{sejourne2019sinkhorn}:
\begin{equation}\label{MMeq}
\inf_{\gamma\in\mathcal{M}^+\left(\mathcal{X}\right)} \int_{\mathcal{X}} c(x)\dd\gamma(x) + \varepsilon KL(\gamma || \bigotimes_{j=1}^J\mu_j) + \sum_{j=1}^J D_\phi\left(\pi_j\gamma\vert\vert \mu_j\right).
\end{equation}
Here, $\mathcal{M}^+\left(\mathcal{X}\right)$ denotes the set of all non-negative finite measures on $\mathcal{X}\subset \left(\mathbb{R}^d\right)^J$, $\bigotimes_j\mu_j(x) = \mu_1(x_1)\otimes\cdots\otimes\mu_J(x_J)$ denotes the independence coupling of the marginal measures $\mu_1,\ldots,\mu_J$,
\[KL(\gamma||\bigotimes_j\mu_j)\coloneqq \begin{cases} \int_{\mathcal{X}} \ln\frac{\dd\gamma}{\dd\bigotimes_j\mu_j}(x)\dd\gamma(x) + (\int_{\mathcal{X}}\dd\bigotimes_j\mu_j(x) - \int_{\mathcal{X}}\dd\gamma(x)) & \text{if $\gamma\ll\bigotimes_j\mu_j$}\\ +\infty&\text{otherwise}\end{cases}\] denotes the Kullback-Leibler divergence \citep{kullback1951information, kullback1959information}, $\pi_j\gamma$ denotes the projection onto the $j$-th marginal of the coupling $\gamma$, and $D_\varphi$ denotes the $\phi$-divergence (or Csisz\`ar-divergence) \citep{,csiszar1967information, csiszar1975divergence}
\[D_\phi\left(\mu||\nu\right)\coloneqq \int_{\mathcal{X}} \phi\left(\frac{\dd\mu}{\dd\nu}\right)\dd\nu + \phi'_\infty \int_{\mathcal{X}}\dd\mu^\perp,\]
where $\phi:[0,\infty)\to[0,\infty]$ is an entropy function, which is assumed to be convex, lower semi-continuous, and $\phi(1)=0$. $\phi'_\infty$ is the recession constant and $\mu^\perp$ is the singular part of $\mu$ with respect to $\nu$ in the Lebesgue decomposition. By assumption, if $D_\phi\left(\mu||\nu\right)<+\infty$ then $\mu\ll \nu$. 

We phrased the unbalanced optimal transport problem as a multimarginal problem, but in most practical settings it will be used in the binary setting, which is of the form
\begin{equation}\label{MMeq_bivariate}
\inf_{\gamma\in\mathcal{M}^+\left(\mathcal{X}\times\mathcal{Y}\right)} \int_{\mathcal{X}\times\mathcal{Y}} c(x,y)\dd\gamma(x,y) + \varepsilon KL(\gamma || \mu\otimes\nu) + D_\phi\left(\pi_1\gamma\vert\vert \mu\right) + D_\phi\left(\pi_2\gamma\vert\vert \nu\right).
\end{equation}

The main difference to the classical optimal transport problem are (i) that the optimal coupling $\gamma$ does not need to be a probability measure and (ii) the addition of the $\phi$-divergence terms, which allow for the creation and destruction of mass and do not enforce the constraint from classical optimal transport that the corresponding measures $\mu$ and $\nu$ need to have the same mass. In practice, this means that computing optimal couplings via \eqref{MMeq_bivariate} or its multimarginal analogue \eqref{MMeq} will provide partial optimal matches where individuals that do not have a close enough match are automatically discarded. This is what provides a small bias in finite samples. 

From a mathematical perspective, the unbalanced problem is a generalization of the entropy-regularized optimal transport problem: setting $D_\varphi = +\infty$ if $\pi_j\gamma \neq \mu_j$ almost everywhere and $0$ otherwise and fixing $\varepsilon>0$ reduces our setting to this balanced problem. The regularity properties of this classical balanced entropy regularized optimal transport problem have been analyzed thoroughly, see for instance \citet{carlier2020differential, carlier2021linear, genevay2019sample, di2020optimal} for recent results and references therein. Note that as $\varepsilon\to0$ the optimal matching $\gamma^*$ solving \eqref{MMeq} approximates the optimal matching obtained via the classical Kantorovich problem \eqref{eq:classic} under mild regularity assumptions \citep{nutz2021entropic, eckstein2021quantitative}.

\subsection{The dual problem and a generalization of the IPFP/Sinkhorn algorithm}
The utility of the unbalanced optimal transport formulation in our matching setup stems from its dual representation, which follows from the classical Fenchel-Rockafellar theorem \citep{chizat2018scaling, sejourne2019sinkhorn}:
\begin{multline}\label{eq:dual}
\sup_{\varphi\coloneqq (\varphi_1,\ldots,\varphi_J)\in \prod_{j=1}^J C\left(\mathcal{X}_j\right)} F(\varphi)\coloneqq\\ -\sum_{j=1}^J\int_{\mathcal{X}_j}\phi^*\left(-\varphi_j(x_j)\right)\dd\mu_j(x_j) -\varepsilon \int_{\mathcal{X}}\left[\exp\left(\frac{\sum_{j=1}^J\varphi_j(x_j)-c(x)}{\varepsilon}\right)-1\right]\dd\bigotimes_{j=1}^J\mu_j(x),
\end{multline}
where $C(\mathcal{X})$ is the space of all continuous functions on $\mathcal{X}$. 

The reason for the utility is that the first-order optimality conditions to \eqref{eq:dual} can be derived by Moreau's decomposition in Banach spaces \citep[Theorem 1]{combettes2013moreau}, which leads to the following generalized Sinkhorn system \citep[Proposition 7]{sejourne2019sinkhorn}:
\begin{equation}\label{eq:FOC}
\varphi_j(x_j) = -\text{aprox}_{\phi^*}^\varepsilon\left(\varepsilon\log\int_{\mathcal{X}_{-j}} \exp\left(\frac{\sum_{i\neq j}\varphi_i(x_i) - c(x)}{\varepsilon}\right)\dd\bigotimes_{i\neq j}\mu_i(x_i)\right)\qquad\forall j=1,\ldots, J,
\end{equation}
where 
\[\text{aprox}_{\phi^*}^\varepsilon(p)\coloneqq \argmin_{q\in\mathbb{R}}\varepsilon\exp\left(\frac{p-q}{\varepsilon}\right) +\phi^*(q) \qquad\text{for all $p\in\mathbb{R}$}\] is the anisotropic approximation operator \citep{combettes2013moreau, sejourne2019sinkhorn}, and where
\[\phi^*(q)\coloneqq \sup_{p\geq 0} pq-\phi(p)\] is the classical Legendre-Fenchel transform of $\phi(p)$ \citep{rockafellar1997convex}. We write as a shorthand
\begin{multline*}\int_{\mathcal{X}_{-j}} f(x_{-j})\dd\bigotimes_{i\neq j} \mu_i(x_i) \\= \int_{\mathcal{X}_1}\cdots\int_{\mathcal{X}_{j-1}}\int_{\mathcal{X}_{j+1}}\cdots \int_{\mathcal{X}_J}f(x_1,\ldots,x_{j-1},x_{j+1},\ldots,x_J)\dd\mu_1\otimes\cdots\otimes\mu_{j-1}\otimes\mu_{j+1}\cdots\otimes \mu_J.
\end{multline*}

The generalized Sinkhorn system \eqref{eq:FOC} can be solved by a generalization of the classical Sinkhorn algorithm/IPFP \citep{deming1940least, sinkhorn1964relationship} as in \cite{sejourne2019sinkhorn}, which is computationally efficient in practice by recursively solving for the sequences $\left\{\varphi_j^m\right\}_{m\in\mathbb{N}\cup\{0\}}$, $j=1\ldots,J$ defined via
\begin{equation}\label{eq:sinkhornconv}
\begin{aligned}
\varphi_j^{m+1}(x_j) &= -\text{aprox}_{\phi^*}^\varepsilon\left(\varepsilon\log\int_{\mathcal{X}_{-j}} \exp\left(\frac{\sum_{i<j}\varphi_i^{m+1}(x_i) + \sum_{i>j}\varphi_i^{m}(x_i)- c(x)}{\varepsilon}\right)\dd\bigotimes_{i\neq j}\mu_i(x_i)\right)
\end{aligned}
\end{equation}
for all $j=1,\ldots, J$. Note that in general the optimal potentials $\{\varphi^*_j\}_{j=1}^J$ need not be unique. However, under the assumption that $\phi^*$ is strictly convex, the dual functional \eqref{eq:dual} becomes strictly convex in $(\varphi_1,\ldots,\varphi_J)$, which implies a unique solution \citep{sejourne2019sinkhorn}. We uphold this assumption for our statistical results below (Assumption \ref{ass:div}). 

In the two-marginal setting, \citet[Theorem 1]{sejourne2019sinkhorn} showed that this approach converges to the solution of \eqref{eq:dual} as $m\to\infty$ for Lipschitz cost functions under some regularity assumptions on the function $\phi^*$ and a normalization of the potentials, which we also uphold (Assumption \ref{ass:div}). The same argument using compactness also holds in the multimarginal setting on compact supports with Lipschitz cost functions. Addressing related questions, \citet{di2020optimal} and \citet{carlier2021linear} proved (linear) convergence of the sequence $\varphi_j^m$ to the optimal $\varphi^*\equiv(\varphi^*_1,\ldots,\varphi^*_J)$ as $m\to\infty$ in the special case of the multimarginal extension of the classical Sinkhorn algorithm for different degrees of smoothness of the cost function, extending convergence results from the two-marginal case for the classical Sinkhorn algorithm \citep{chen2016entropic, franklin1989scaling, ruschendorf1995convergence}.

This dual approach computes the optimal potentials $\varphi^*$, and only up to the additive constants $\lambda_{\varphi_j}$, which are fixed by normalization. From these optimal potentials, one can generate the optimal matching $\gamma^*$ via 
\begin{equation}\label{eq:optim_measure}
\dd\gamma^* = \exp\left(\frac{\sum_{j=1}^J\varphi_j^*-c}{\varepsilon}\right)\cdot\dd\bigotimes_j\mu_j,
\end{equation} 
i.e.~the optimal potentials are used to generate the optimal matching $\gamma$ as an exponential tilting of the independence coupling of the marginals $\mu_j$. $\gamma^*$ and $\varphi^*_j$ all depend on $\varepsilon$, but we make this dependence implicit for notational convenience.

It is the IPFP that makes solving the entropy regularized problem so convenient. The above procedure \eqref{eq:sinkhornconv} is written out for general measures $\mu_j$. In practice, where the measures $\mu_j$ are unknown and have to be estimated from observations, efficient implementations have been proposed, see \citet{peyre2019computational} and references therein. Of particular relevance for high-dimensional matching scenarios are mini-batch approaches that break the problem into smaller parts \citep{fatras2021unbalanced}.

For the multimarginal case introduced above, the optimal probabilistic matching $\gamma^*$ of the participants is joint over all $J$ treatment arms. In particular, for each individual $n_j$ in treatment arm $j$ it provides the optimally weighted matchings in each other treatment arm $i\neq j$. The additional information obtained from the full multimarginal matching comes with an increase in computational cost.  That is, one needs to specify (and store in memory) the squared distances between all possible points of all $J$ marginal measures. This becomes increasingly costly with the number of arms $J$, even for the most efficient implementations \citep{pham2020unbalanced}. In addition, to obtain a tensor joint distribution that can provide a mapping close to a curve in multi-dimensional space, instead of a more uniformly spread-out mass, one needs to specify a penalty term $\varepsilon>0$ that is much smaller than that in the pairwise setting, which not only can make it difficult for the algorithm to converge, but also bring stability issues.

Therefore, we recommend a pairwise comparisons approach to reduce the computational burden in very high-dimensional settings with many treatment arms. Still, we phrase and prove all results in the general multimarginal setting.

\section{Causal effects via optimal transport matching}\label{sec:ATE}
This section introduces matching via optimal transportation as a natural framework to minimize bias in weighting approaches for estimating average treatment effects. We rely on the potential outcome notation \citep{rubin1974estimating}. Let $Y(j)$ be the potential outcome under the causal inference framework for treatment $j\in \{1,...,J\}$, $T$ be the treatment assignment, and random variable $D$ supported on $\{\tau\in \{0,1\}^J: \sum_{j=1}^J \tau_j=1\}$ be the binary vector of the treatment assignment: $T=j$ if and only if $D(j)=1$. $X\in\mathbb{R}^d$ is the vector of observed covariates. 

We make the following set of assumptions, the first three standard in the literature \citep{abadie2006large, imbens2000role, imbens2004nonparametric}.
\begin{assumption}\label{ass:causal}
For $j=1,..,J$,
\begin{enumerate}
    \item (Weak unconfoundedness) $D(j) \independent Y(j)|X$,
    \item (Overlap) $P(D(j)=1|X)$ is bounded away from 0 and 1,
    \item Conditional on $D(j)=1$, the sample is independently drawn from $(Y,X)|D$, and the marginal $X\sim \mu_j$ given $D(j)=1$,
    \item Given $D(j)=1$, there exists a bounded and continuous function $f_j$ supported on a convex and compact set $\mathcal{X}_j\subset \mathbb{R}^d$ such that $f_j(X) = \mathbb{E}(Y|T=j,X)$. For a more explicit notation, we use $y(j,X):=f_j(X)$ and $\varepsilon_i = Y_i-y(T_i,X_i)$ are i.i.d. mean-zero noise terms.
\end{enumerate}
\end{assumption}
Parts 1 and 2 of Assumption \ref{ass:causal} are standard. The weak unconfoundedness assumption formalizes the idea of matching on observables: it guarantees that there are no more systematic differences between the unobservable covariates of the groups once the observable characteristics have been accounted for. The weak unconfoundedness assumption was introduced in \citet{imbens2000role} and is slightly weaker than the strong unconfoundedness assumption introduced in \citet{rosenbaum1983central} since it is locally constrained for every realization of the treatment assignment. The overlap assumption (Part 2 of Assumption \ref{ass:causal}) guarantees that one can find appropriate matches in each group. If this were not the case, the estimate of the causal effect would be biased. It is a fundamental assumption for any matching approach \citep{stuart2010matching}. Part 3 of Assumption \ref{ass:causal} states that all treatment groups are independently sampled; a similar assumption is made in classical matching approaches \citep{abadie2006large}. Part 4 introduces regularity assumptions that simplify the derivations of our statistical results; in particular, the boundedness assumption can be relaxed. In general, these regularity assumptions are weaker than most existing assumptions in the literature, as we do not specify any parametric restrictions on the relationship between $Y$ and $X$ like linearity or convexity and only require additive separability analogously to \citet{abadie2006large}. 

Under part 3 in Assumption \ref{ass:causal}, we can denote $X_j := X|D(j)=1$, $X_j\sim \mu_j$ and $\gamma$ as the joint coupling over $\prod_{j=1}^J\mathcal{X}_j$ with the marginal distributions $(\mu_1,...,\mu_J)$ obtained via the proposed matching procedure. The expected counterfactual outcome $\mathbb{E}[Y(t)]$ for treatment $T=t$ can then be written as $\mathbb{E}[Y(t)] = \mathbb{E}_{X}\left[ \mathbb{E}(Y(t)|D(t)=1,X) \right] = \mathbb{E}_X[y(t,X)]$.

The idea of the counterfactual notation is that one can only observe one counterfactual outcome $Y(j)$ in realized samples. This is the ``fundamental problem of causal inference'' \citep{holland1986statistics}: for each individual $i=1,\ldots, N$ the researcher only observes one potential outcome $Y_i(j)$ for the given treatment $T=j$ that the individual was assigned to. In order to compute the causal effect of a treatment, the researcher should compare $Y_i(j)$ to $Y_i(k)$, $k\neq j$ for each $i$, which is impossible. The idea of matching on covariates is hence to impute the counterfactual outcomes $Y_i(k)$ for individual $i$ by considering the outcomes of other individuals that received a different treatment $k$. The unconfoundedness assumption implies that individuals with similar observable covariates also possess similar unobservable characteristics, which makes matching on observable covariates a feasible approach.

This leads to the following definition of the potential outcome imputed by our proposed method.

\begin{definition}\label{def:matching_po1}
For the $j$-th treatment and $t\neq j$ denote $\gamma_{j|t}$ as the conditional measure of covariates in group $j$ given the covariates in group $t$ under the joint distribution $\gamma$. 
Then under Assumption \ref{ass:causal}, the expected potential outcome can be expressed in the sample version as
\begin{equation}\label{eq:sample_po}
    \hat{\mathbb{E}}_N \left[\hat{Y}(j)\right] = \frac{1}{N}\sum_{i=1}^N Y_i I(D_i(j)=1) + 
    \frac{1}{N}\sum_{i=1}^N \sum_{t\neq j} \sum_{k\neq i} Y_k \hat{\gamma}_{N,j|t}(X_k|X_i)I(D_i(t)=1),
\end{equation}
where $N = \sum_{j=1}^J N_j$ is the overall number of sample points over all treatment arms and $\hat{\gamma}_N$ is the empirical counterpart to the optimal matching estimated via the generalized Sinkhorn algorithm \eqref{eq:sinkhornconv} by replacing $\mu_j$ with the empirical measures $\hat{\mu}_{N_j}$ defined below in \eqref{eq:empmeas}.
\end{definition}
The above definition is designed for the multiple treatment setting. For binary treatments, the expression \eqref{eq:sample_po} is the empirical analogue of \[\mathbb{E}\left[\hat{Y}(1)\right] = \mathbb{E}\left[\mathbb{E}\left[\hat{Y}(1)|T=1,X\right] + \mathbb{E}\left[\hat{Y}(1)|T=0,X\right]\right].\] It is derived by replacing the population measures $\mu_j$ by empirical measures. From this, the average treatment effect or average treatment effects on the treated can be obtained. In the same manner, we can provide estimators of average treatment effects of subgroups like the average treatment effect on the treated, or conditional average treatment effects. We now provide the asymptotic statistical properties of this estimator. 

\section{Asymptotic statistical properties of the optimal matching}\label{sec:stats}
We now provide asymptotic statistical properties of the estimator $\hat{\gamma}_N$ of the optimal matching $\gamma$ obtained by solving \eqref{eq:FOC} and using the characterization \eqref{eq:optim_measure}. As mentioned, estimating the weights $\gamma$ via the classical optimal transport problem is known to suffer from a statistical curse of dimensionality. However, we show below that for fixed $\varepsilon>0$, the unbalanced optimal transport problem does not suffer from this curse and obeys a central limit theorem at the parametric rate. This makes it convenient for practical applications, especially in high-dimensional settings.

The results in these section complement the analytical and statistical results in \citet{sejourne2019sinkhorn}, who \emph{inter alia} proved that the statistical estimator of the value function of the unbalanced optimal transport problem is consistent at the parametric rate. We show that even the empirical process of the optimizer, i.e.~the optimal matching, converges at the parametric rate to a tight limit process with finite (but potentially degenerate) variance if we replace the population measures $\mu_j$ by their empirical counterparts. This result also complements the recently introduced asymptotic linearity result in \citet{harchaoui2020asymptotics} for the classical balanced entropy-regularized optimal transport problem. There, the authors focus on the two-treatment setting (but their results can be extended to the multimarginal setting) under the assumption that both treatment arms contain the same number of observations. In contrast to our result, they also derive expressions for the asymptotic variance via Gaussian chaos expansions. \citet{klatt2020empirical}, also in the binary setting, show asymptotic normality of the optimal matching of the classical entropy regularized optimal transport matching and validity of the bootstrap in the setting where all distributions are purely discrete. So far, analogous results for the unbalanced optimal problem have been absent and the results in this section are designed to take a first step to fill this gap. 

To set the notation, we define the empirical measures 
\begin{equation}\label{eq:empmeas}
    \hat{\mu}_{N_j} \coloneqq \frac{1}{N_j}\sum_{n_j=1}^{N_j} \delta_{X_{n_j}}
\end{equation} 
for $j=1,\ldots, J$, where $\delta_{X_{n_j}}$ denotes the Dirac measure at the observation $X_{n_j}$. This allows for differing sample sizes $N_j$ in each cohort. 

Throughout, we assume that we observe independent and identically distributed draws $X_{n_j}$ from the respect distributions $\mu_j$. For notational convenience we will index estimators of functional relationships $f$ that depend on all $j$ elements as $\hat{f}_N$, where it is implicitly understood that $N\equiv \bigcup_{j=1}^J N_j$ is the overall number of data points across all groups. We require the following assumptions for our statistical analysis.
\begin{assumption}\label{ass:iid}
For all $j=1,\ldots, J$ the observable random variables $\left\{X_{n_j}\right\}_{n_j=1,\ldots,N_j}$ are independent and identically distributed draws from the measure $\mu_j$.
\end{assumption}
\begin{assumption}\label{ass:support}
For all $j=1,\ldots, J$ the supports $\mathcal{X}_j\subset\mathbb{R}^d$ of $\mu_j$ are convex and compact.
\end{assumption}
\begin{assumption}\label{ass:asybal}
The samples are asymptotically balanced, i.e.~there exist constants $\rho_1,\rho_2,...,\rho_J>0$ such that $\frac{N_j}{\sum_j N_j} \to \rho_j$ as $N_j\to\infty$ for all $j=1,\ldots, J$.
\end{assumption}
To ensure the regularity and uniqueness of the optimal potentials in the problem, we assume more structure on the function $\phi$ in the definition of the divergence $D_\phi$. The following assumption is satisfied by many functions used for divergences, such as the Kullback-Leibler- and power entropies \citep{sejourne2019sinkhorn}. For other forms of divergences such as the total variation or even the classical balanced entropy regularized optimal transport problem, uniqueness of the potentials is no longer guaranteed \citep[section 3.3.4]{sejourne2019sinkhorn}, but similar results can be obtained by different proof methods.
\begin{assumption}\label{ass:div}
$\phi^*$, the convex conjugate of $\phi$, is strictly convex and infinitely often differentiable. Moreover, there exists a sequence $\{y_k\}_{k=1}^\infty$ in the domain of $\phi^*$ such that $\frac{\dd}{\dd y}\phi^*(y_k)$ converges either to $+\infty$ or zero.
\end{assumption}

\begin{assumption}\label{ass:cost}
The cost function $c:\left(\mathcal{X}_1,\ldots,\mathcal{X}_j\right)\to\mathbb{R}^+$ is infinitely often differentiable. 
\end{assumption}
The strong regularity on the cost function is important as the regularity transfers onto the optimal potentials. The smoothness of the potentials is crucial for a parametric rate of convergence of the optimal couplings $\gamma^*_{N}$, as it provides a Donsker class for all dimensions $d$ \citep{wellner2013weak}. A standard example for a cost function satisfying Assumption \ref{ass:cost} is the squared Euclidean distance function, which is the most natural for a causal matching approach.  

\citet{sejourne2019sinkhorn} have shown that the generalized Sinkhorn algorithm solves the unbalanced optimal transport problem and that the potentials vary continuously with the input measures. We repeat the results for the consistency and convergence of the generalized Sinkhorn algorithm from \citet{sejourne2019sinkhorn} here in our setting. The proofs are analogous to theirs, as their arguments can be directly extended to the multimarginal setting. 

\begin{proposition}\label{prop:consistency}
Under Assumptions \ref{ass:iid} -- \ref{ass:cost}, for fixed measures $\{\mu_{j}\}_{j=1}^J$, the sequence of potentials $\{\hat{\varphi}_j^m\}_{j=1}^J$ in the generalized Sinkhorn algorithm \eqref{eq:sinkhornconv} converges to the true optimal potentials $\{\varphi_j^*\}_{j=1}^J$ as $m\to\infty$. Moreover, the optimal potentials $\{\hat{\varphi}_{N_j}^*\}_{j=1}^J$ solving the empirical analogues of the generalized Sinkhorn algorithm \eqref{eq:sinkhornconv}, where the measures $\mu_j$ are replaced by their empirical counterparts $\hat{\mu}_{N_j}$, converge uniformly in probability to $\{\varphi^*_j\}_{j=1}^J$. Furthermore, the optimal empirical couplings $\hat{\gamma}_{N}^*$ converge weakly in probability to $\gamma^*$.
\end{proposition}

We now provide the main results of this section: a statistical analysis of the asymptotic properties of the optimal matching generated by the generalized Sinkhorn algorithm. To simplify the subsequent statements, we define the operator $T(\varphi)\coloneqq (T_1(\varphi),\ldots, T_J(\varphi))$ with
\[T_j(\varphi)(x_j) = \varphi_j(x_j) - \text{aprox}_{\phi^*}^\varepsilon \left(\varepsilon\log\int_{\mathcal{X}_{-j}}\exp\left(\frac{\sum_{i\neq j}\varphi_i(x_i)-c(x)}{\varepsilon}\right)\dd\bigotimes_{i\neq j}\mu_i(x_{-j})\right),\] which corresponds to the generalized Sinkhorn setting \eqref{eq:FOC}. We take its domain to be $\prod_{j=1}^J C(\mathcal{X}_j)$, which is a Banach space under the supremum norm \[\|\varphi\|_{\infty} \coloneqq \sum_{j=1}^J \|\varphi_j\|_{\infty} = \sum_{j=1}^J\sup_{x_j\in\mathcal{X}_j}  \left\lvert \varphi_j(x_j)\right\rvert. \]
We also denote $\hat T_N(\varphi) \equiv \left(\hat{T}_{N_1}(\varphi),\ldots,\hat{T}_{N_J}(\varphi)\right)$ with \[\hat{T}_{N_j}(\varphi) = \varphi_j(x_j) - \text{aprox}_{\phi^*}^\varepsilon \left(\varepsilon\log\int_{\mathcal{X}_{-j}}\exp\left(\frac{\sum_{i\neq j}\varphi_i(x_i)-c(x)}{\varepsilon}\right)\dd\bigotimes_{i\neq j}\hat{\mu}_{N_i}(x_{-j})\right).\]

As a shorthand, and following \cite{sejourne2019sinkhorn}, we write
\[T_j(\varphi)(x_j) = \varphi_j(x_j) -  \text{aprox}_{\phi^*}^\varepsilon\left(-\text{Smin}_j^\varepsilon\left(c-\sum_{i\neq j}\varphi_i\right)\right),\] where
\[\text{Smin}_j^\varepsilon\left(c-\sum_{i\neq j}\varphi_i\right) \coloneqq -\varepsilon\log\int_{\mathcal{X}_{-j}}\exp\left(\frac{\sum_{i\neq j}\varphi_i-c}{\varepsilon}\right)\dd\bigotimes_{i\neq j}\mu_i\] is the ``soft-min'' operator. Based on these definitions we have the following result for the asymptotic distribution of the optimal empirical potentials $\hat{\varphi}^m_{N_j}$ obtained in the $m$-th iteration of the IPFP \eqref{eq:sinkhornconv}. 

\begin{theorem}\label{thm:potential}
Let Assumptions \ref{ass:iid} -- \ref{ass:cost} hold. Then for every $(N_1,\ldots, N_J)$ there exists a number of iterations $m\in\mathbb{N}$ in the IPFP \eqref{eq:sinkhornconv} such that
\[\left(\sqrt{N_1}\left(\hat{\varphi}^m_{N_1} - \varphi_1^*\right),\ldots, \sqrt{N_J}\left(\hat{\varphi}^m_{N_J}-\varphi_J^*\right)\right) \rightsquigarrow - \left(T'(\varphi^*)\right)^{-1}(Z),\] 
where $\rightsquigarrow$ denotes weak convergence and $Z$ is the tight mean-zero Gaussian limit element of the form
\[\left(\sqrt{N_1} \left(\hat{T}_{N_1}(\varphi^*) - T_1(\varphi^*)\right),\ldots, \sqrt{N_J} \left(\hat{T}_{N_J}(\varphi^*) - T_J(\varphi^*)\right)\right) \]
in $\ell^\infty\left(B_1\left(\prod_jC^\infty_c(\mathcal{X}_j)\right)\right)$, the space of all bounded functions on the unit ball of all smooth functions $v\equiv (v_1,\ldots, v_J)$ compactly supported $\prod_j\mathcal{X}_j$. The bounded linear operator $T'(\varphi)(u)\coloneqq \left(T'_j(\varphi)(u),\ldots, T_J'(\varphi)(u)\right): \prod_j C(\mathcal{X}_j) \to \prod_jC(\mathcal{X}_j)$ with
\[T'_j(\varphi)(u) = u_j +\varepsilon Y\cdot\frac{\int_{\mathcal{X}_{-j}}\exp\left(\frac{\sum_{i\neq j}\varphi_i-c}{\varepsilon}\right)\sum_{i\neq j}u_i\dd\bigotimes_{i\neq j}\mu_i}{\int_{\mathcal{X}_{-j}}\exp\left(\frac{\sum_{i\neq j}\varphi_i-c}{\varepsilon}\right)\dd\bigotimes_{i\neq j}\mu_i},\] where
\[ Y \equiv \frac{{\phi^*}'\left[\mathrm{aprox}_{\phi^*}^\varepsilon\left(-\mathrm{Smin}_j^\varepsilon\left(c-\sum_{i\neq j}\varphi_i\right)\right)\right]}{{\phi^*}'\left[\mathrm{aprox}_{\phi^*}^\varepsilon\left(-\mathrm{Smin}_j^\varepsilon\left(c-\sum_{i\neq j}\varphi_i\right)\right)\right]+ \varepsilon {\phi^*}''\left[\mathrm{aprox}_{\phi^*}^\varepsilon\left(-\mathrm{Smin}_j^\varepsilon\left(c-\sum_{i\neq j}\varphi_i\right)\right)\right]}\] 
is the Fr\'echet-derivative, $(T'(\varphi))^{-1}(u)$ is its inverse, and ${\phi^*}''$ is the second derivative of $\phi^*$. The inverse exists and is bounded for $0<\varepsilon<1 $ if we equip the image space of $T'(\varphi)(\cdot)$ with the product norm $\|v\|_m\coloneqq \max_{j} \|v_j\|_\infty$. The classical bootstrap procedure is valid.  
\end{theorem}
The number of iterations $m$ needs to be chosen such that 
\[\hat{T}_{N_j}(\hat{\varphi}^m_N) = o_P(N^{-\frac{1}{2}}),\] where $o_P$ is the standard Landau symbol and we always understand statements about probabilities in terms of outer probability \citep[Section 1.2]{wellner2013weak}. Theorem \ref{thm:potential} shows that the potential functions obtained via the IPFP possess excellent statistical properties. In particular, the proof of Theorem \ref{thm:potential} also shows that a standard bootstrap procedure, i.e.~a resampling procedure of the observable data with replacement, provides correct results in the sense that the bootstrapped estimator $\tilde{\varphi}_N^{m}$ satisfies 
\[\left(\sqrt{N_1}\left(\tilde{\varphi}^m_{N_1} - \hat{\varphi}^m_{N_1}\right),\ldots, \sqrt{N_J}\left(\tilde{\varphi}^m_{N_J}-\hat{\varphi}^m_{N_J}\right)\right) \] weakly converges to a tight Gaussian limit process conditional on the observed data like the standard empirical process
\[\left(\sqrt{N_1}\left(\hat{\varphi}^m_{N_1} - \varphi_1^*\right),\ldots, \sqrt{N_J}\left(\hat{\varphi}^m_{N_J}-\varphi_J^*\right)\right).\] This follows from Theorem 13.4 in \citet{kosorok2008introduction}. We refer to \citet{kosorok2008introduction} and chapter 3.6 in \citet{wellner2013weak} for a detailed analysis of statistical properties of different bootstrap procedures.

Theorem \ref{thm:potential} is the main step in the asymptotic properties of the optimal matching $\gamma^*_\varepsilon$. Recall that the optimal matching $\gamma^*_\varepsilon$ can be derived from the potentials $\varphi^*$ via 
   \[ \gamma^*_\varepsilon(A) = \int_{A}\exp\left(\frac{\sum_{J=1}^J\varphi^*_j(x_j) - c(x)}{\varepsilon}\right)\dd\bigotimes_j\mu_j(x)\]
for any measurable set $A$, i.e.
\[K^*(x)\coloneqq \exp\left(\frac{\sum_{J=1}^J\varphi^*_j(x_j) - c(x)}{\varepsilon}\right)\] is the Radon-Nikodym derivative of $\gamma^*_\varepsilon$ with respect to $\bigotimes_j\mu_j$. The asymptotic behavior of the Radon-Nikodym derivative is straightforward to obtain via the delta method \citep[chapter 3.9]{wellner2013weak} as in the following proposition.
\begin{corollary} \label{corr:density}
    Let Assumptions \ref{ass:iid} -- \ref{ass:cost} hold. Then for every $(N_1,\ldots, N_J)$ with $N_j=\rho_jN$ there exists a number of iterations $m\in\mathbb{N}$ in the IPFP \eqref{eq:sinkhornconv} such that
    \begin{multline*}
        \sqrt{N} \left(\exp\left(\frac{\sum_{J=1}^J\hat{\varphi}_{N_j}^m - c}{\varepsilon}\right) - \exp\left(\frac{\sum_{J=1}^J\varphi^*_j - c}{\varepsilon}\right)\right) \rightsquigarrow
        -\exp\left( \frac{\sum_j\varphi^*_j-c}{\varepsilon} \right) \frac{\sum_{j=1}^J \rho_j^{-1}Z_{\varphi,j}}{\varepsilon}
    \end{multline*}
    where 
    \[Z_\varphi\equiv(Z_{\varphi,1},\ldots,Z_{\varphi,J})\coloneqq -\left(T'(\varphi^*)\right)^{-1}(Z),\] $Z$ is the mean zero Gaussian limit process from Theorem \ref{thm:potential}. The limit process is still mean zero Gaussian as it is a linear map of $Z_\varphi$.
\end{corollary}
Corollary \ref{corr:density} shows that the exponential tilt imposed by $K^*(x)$ also converges at the parametric rate to a Gaussian process. Hence, if all measures $\mu_j$ are known, this implies that the optimal coupling converges at the parametric rate to a Gaussian process. Moreover, the Delta method for bootstrapped measures \citep[section 3.9.3]{wellner2013weak} implies that the standard bootstrap procedure works in this case. 

Corollary \ref{corr:density} provides the statistical properties of the optimal matching $\gamma^*_\varepsilon$ in the case where $\mu_j$ are known. In general, this is not the case and we need to take into account the estimation of the measures $\mu_j$. The following theorem shows that we still have convergence to a tight mean-zero limit process in this case. Also, the classical bootstrap is still valid.
\begin{corollary} \label{corr:matching}
Let Assumptions \ref{ass:iid} -- \ref{ass:cost} hold. Denote 
\[K^*(x)\coloneqq \exp\left(\frac{\sum_{j=1}^J\varphi^*_j(x_j) - c(x)}{\varepsilon}\right)\quad\text{and}\quad \hat K_N(x)\coloneqq \exp\left(\frac{\sum_{j=1}^J\hat\varphi_{N_j}(x_j) - c(x)}{\varepsilon}\right)\]

Denote $\dd\hat\gamma_{N,\varepsilon}(x):=\hat K_N(x) \dd\bigotimes_{j=1}^J\hat\mu_{N_j}(x)$ and $\dd\gamma_\varepsilon(x):=K^*(x)\dd\bigotimes_{j=1}^J\mu_{j}(x)$.
Then for any function $f$ on $\prod_j\mathcal{X}_j$ in some uniformly bounded Donsker class $\mathcal{F}$ 
\[\sqrt{N}\left( \int f(x) \dd\hat\gamma_{N,\varepsilon}(x) - \int f(x) \dd\gamma_\varepsilon(x)  \right) \rightsquigarrow G\]
where $G$ is a mean-zero tight process.
\end{corollary}
Corollary \ref{corr:density} and Corollary \ref{corr:matching} show that the optimal balancing possesses converges at the parametric rate to the limit matching. Note that Corollary \ref{corr:matching} is weaker than the central limit result in \citet{harchaoui2020asymptotics} for the classical balanced entropy regularized optimal transport problem, who provide a Gaussian limit and compute the variance for the classical balanced entropy regularized optimal transport problem. We conjecture that the limit process in Corollary \ref{corr:matching} will be Gaussian in general, but we have not managed to show this. Moreover, the variance of the limit process can be trivial, in which case one would need a higher-order expansion of the variance functional similar to \citet{harchaoui2020asymptotics} in the classical balanced case. On the other hand, our result allows for different sample sizes in the different treatment arms, while \citet{harchaoui2020asymptotics} rely on a weighted Monge representation for their results, which requires the same number of observations in each treatment arm. 

The tight limit process $G$ is achieved for bounded Donsker classes $\mathcal{F}$, i.e.~functional classes that are ``small enough'' to be estimable with a relatively small number of observations. This is a standard regularity requirement in the theory of empirical processes. Smooth functions, indicator functions and many other functions fall under this category. For a thorough introduction to Donsker classes we refer to \citet{wellner2013weak}. Corollary \ref{corr:matching} implies the following statistical properties for the induced estimator of the average treatment effect between two groups.
\begin{corollary}\label{cor:conditional}
\mbox{}
\begin{enumerate}
    \item[A.] 
    For any $i,j\in \{1,...,J\}$, $i\neq j$, denote $\hat\gamma(x_i|x_j)$ as the sample version of the conditional distribution $\gamma(x_i|x_j)$, then for any  function $f$ in a uniformly bounded Donsker class and for any $x_j \in\mathcal{X}_j$,
\[\sqrt{N}\left(\int f(x)\dd\hat\gamma(x_i|x_j) - \int f(x)\dd\gamma(x_i|x_j) \right)\rightsquigarrow G_{i|j}\]
where $G_{i|j}$ is a mean-zero tight process.
    \item[B.] 
    In addition, in the binary case $T\in\{0,1\}$, for the imputed potential outcome estimator in Definition \ref{def:matching_po1},
    \[\sqrt{N} \left( \mathbb{\hat E}_\varepsilon \{Y(0)|T=1\} - \mathbb{E}_\varepsilon\{ Y(0)|T=1 \} \right) \rightsquigarrow G_0 \]
    Here, $\mathbb{\hat E}_\varepsilon \{Y(0)|T=1\} =\frac{1}{N_1}\sum_{j=1}^{N_1}\sum_{i=1}^{N_0}Y_i(0)\hat\gamma_{\varepsilon,0|1}(X_{i,T=0}|X_{j,T=1})$ is the imputed potential outcome estimator, with subscription $\varepsilon$ to denote the dependence of $\gamma$ on $\varepsilon$. 
    The population mean under OT mapping, $\mathbb{E}_\epsilon[Y(0)|T=1] := \int_{\mathcal{X}_1} \int_{\mathcal{X}_0}f_0(s)\dd\gamma_{\varepsilon,0|1}(s|x) \dd\mu_1(x)$
    is a biased version of $\mathbb{E}(Y(0)|T=1)$ that depends on $\varepsilon$.
    $G_{0}$ is a mean-zero tight process. 
\end{enumerate}

\end{corollary}
Note that in Part B of Corollary \ref{cor:conditional}, under the binary treatment setting, where the optimal coupling $\dd\gamma_\varepsilon(s,x) = K_\varepsilon(s,x) \dd\mu_0(s)\dd\mu_1(x)$, the bias between $\mathbb{E}_\epsilon[Y(0)|T=1]$ and the true counterfactual $\mathbb{E}[Y(0)|T=1]$ is 
\begin{align*}
    \text{Bias} &= \int_{\mathcal{X}_0\times\mathcal{X}_1}[f_0(s)K_\varepsilon(s,x) - f_0(x)]\dd\mu_0(s)\dd\mu_1(x)
\end{align*}
where $Y_i(0)=f_0(X_i)+\eta_i$, $\eta_i$ is an independent noise term as in Assumption \ref{ass:causal} and where we have made the dependence of $K$ on epsilon explicit.

Based on this expression, it is easy to see that if the divergence terms $D_\phi$ vanish quicker than $\varepsilon$ as the latter goes to zero (for instance, when scaling the divergence terms with a parameter $\rho = o(\varepsilon)$), then $K_\varepsilon(s,x)$ converges to a Dirac delta distribution $\delta_x(s)$, which sets the bias to zero. The explanation for this provides further insight into matching approaches: by setting the bias to zero in the limit one only keeps perfect matches. If the distributions of characteristics between groups differ even in the limit, this entails discarding (potentially a significant amount of) mass even in the limiting distribution. This obviously increases the variance of potential estimates, potentially dramatically so.

Classical optimal transport has a balancing constraint, which enforces a match of all the mass in both distributions. This leads to matches that are not perfect and introduces potentially severe bias; on the other hand, the balancing constraint allows the method to use all data points, which in general provides a low variance. 

Unbalanced optimal transport is a middle ground between the two: by introducing a penalty term for the balancing constraint from classical optimal transport via the divergence terms, it allows for the balancing constraint to not be satisfied perfectly, which allows for discardment of bad matches. On the other hand, the balancing constraint enforces the use of most of the available data, which in general yields a lower variance. Therefore, from a causal inference perspective, the unbalanced optimal transport approach balances bias and variance: the more the balancing constraint of the classical optimal transport problem is enforced, the more bad matches are created introducing bias into the solution while reducing the variance, and vice versa.

\section{Simulations and application to the Lalonde data set}\label{sec:illu}
The competitive statistical finite-sample properties of the proposed method, especially compared to the standard $k$-nearest neighbor approach \citep{abadie2004implementing}, can already be appreciated in simple two-dimensional simulation studies. We propose two simulations, one with a region of overlapping individuals from the two groups, and the other one with the treated sample concentrated within the support of the control sample. Figure \ref{plot:case_plot} provides an illustration of the covariate distributions for the two testing cases. To compare the method to the existing $k$-nearest neighbor approach (KNN), we compare the performance of unbalanced optimal transport with inverse propensity score weighting (IPW), and KNN when $k=1$ and when $k=3$, in terms of the ATE and ATT estimation accuracy. 
\subsection{Simulation studies}
Denote the 2-dimensional covariate $\mathbf{X} = (X_1,X_2)$.
Throughout the simulation studies, we sample the outcome from the following distributions. 
\begin{align*}
    Y(0) &\sim \mathcal{N}(-1+X_1\cdot X_2,1)\\
    Y(1) &\sim \mathcal{N}(2 + 2X_1 + X_2,0.5)
\end{align*}
In addition, we sample $n_0=1000$ and $n_1=100$ individuals for the control and treated group respectively.
\begin{figure}[h]
\centering
\includegraphics[width=\textwidth]{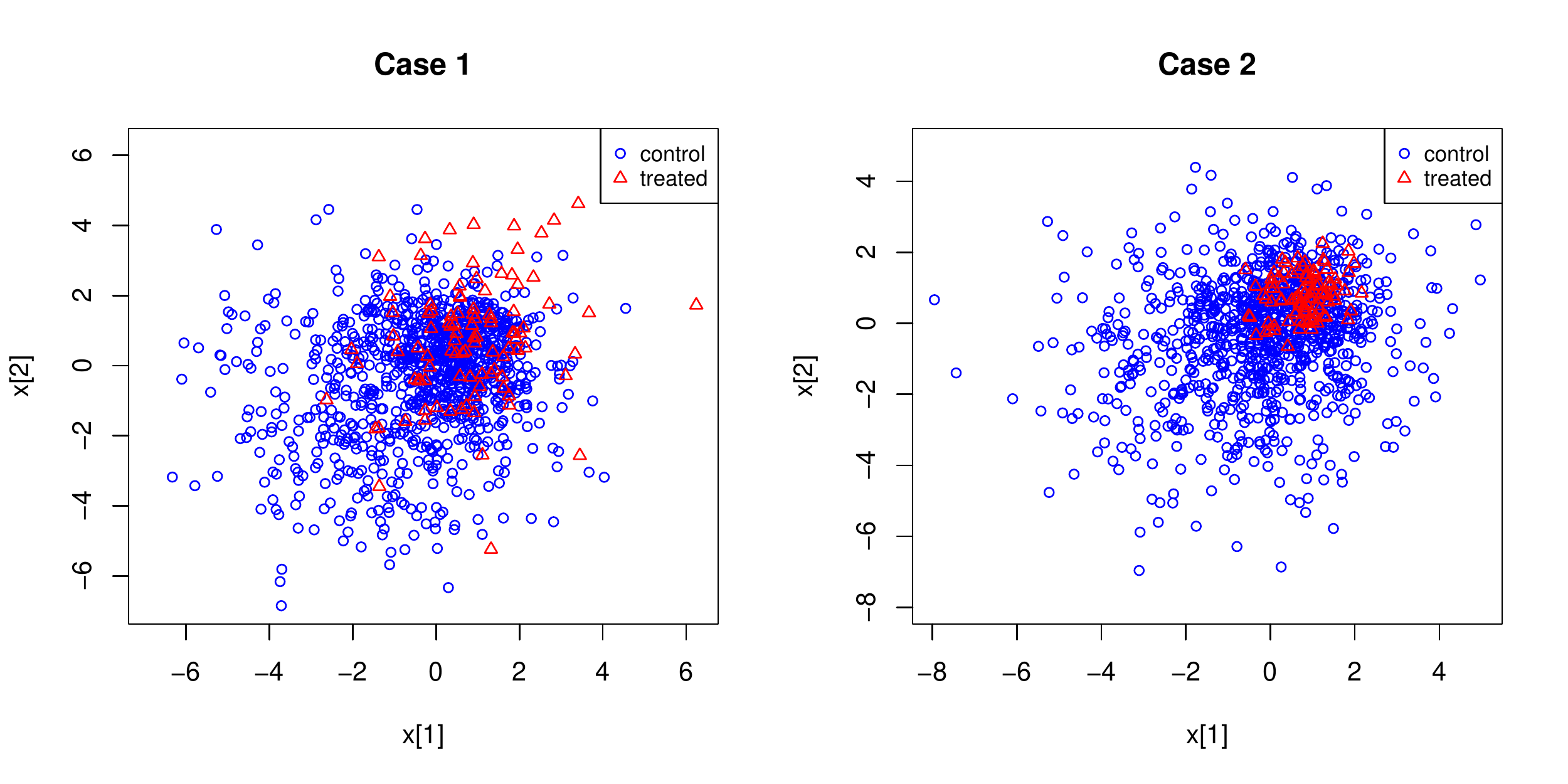}
\caption{Illustration of the covariate distribution for the treated and control groups.}
\label{plot:case_plot}
\end{figure}

\begin{table}[hb!]
\centering
\begin{tabular}{lrrrrrr}
\toprule
\multicolumn{7}{c}{Case 1}                                                                                     \\
                                               & Penalty & OT     & IPW     & KNN(k=3) & KNN(k=1) & Unadjusted \\ \hline
\multirow{4}{*}{ATE\_diff}   & 0.001   & 0.4173 & 20.5054 & 0.6873   & 0.4957   & 2.7452     \\
\multicolumn{1}{c}{}                           & 0.005   & 0.6749 & 20.8503 & 0.6917   & 0.4977   & 2.7471     \\
\multicolumn{1}{c}{}                           & 0.010   & 0.8313 & 20.3398 & 0.6675   & 0.4696   & 2.7474     \\
\multicolumn{1}{c}{}                           & 0.050   & 1.6036 & 20.3374 & 0.6719   & 0.4855   & 2.7309     \\ \hline
\multirow{4}{*}{ATT\_diff}                     & 0.001   & 0.2428 & 1.3931  & 0.2199   & 0.2023   & 0.3068     \\
                                               & 0.005   & 0.2842 & 1.4699  & 0.2450   & 0.2286   & 0.2996     \\
                                               & 0.010   & 0.3275 & 1.5040  & 0.2274   & 0.2234   & 0.2822     \\
                                               & 0.050   & 0.4606 & 1.4734  & 0.2351   & 0.2230   & 0.3172     \\ \hline
\multirow{4}{*}{ATT\_sd\_diff}                 & 0.001   & 0.7943 & -       & 0.9089   & 2.6969   & -          \\
                                               & 0.005   & 0.9530 & -       & 0.9535   & 2.7791   & -          \\
                                               & 0.010   & 1.1830 & -       & 0.9830   & 2.8030   & -          \\
                                               & 0.050   & 1.9420 & -       & 0.9188   & 2.6962   & -          \\ \hline
\multicolumn{7}{c}{Case 2}                                                                                     \\
                                               & Penalty & OT     & IPW     & KNN(k=3) & KNN(k=1) & Unadjusted \\ \hline
\multirow{4}{*}{ATE\_diff}                     & 0.001   & 1.5778 & 30.1113 & 1.6580   & 1.5329   & 2.7020     \\
                                               & 0.005   & 1.4981 & 30.2217 & 1.6696   & 1.5651   & 2.6984     \\
                                               & 0.010   & 1.7395 & 29.9517 & 1.7046   & 1.5796   & 2.7575     \\
                                               & 0.050   & 2.1768 & 29.9721 & 1.6663   & 1.5544   & 2.7252     \\ \hline
\multirow{4}{*}{ATT\_diff}                     & 0.001   & 0.1154 & 1.4290  & 0.1157   & 0.1285   & 0.1240     \\
                                               & 0.005   & 0.1504 & 1.4388  & 0.1057   & 0.1306   & 0.1390     \\
                                               & 0.010   & 0.2160 & 1.4316  & 0.1146   & 0.1301   & 0.1218     \\
                                               & 0.050   & 0.4460 & 1.4361  & 0.1085   & 0.1283   & 0.1320     \\ \hline
\multirow{4}{*}{ATT\_sd\_diff}                 & 0.001   & 0.1630 & -       & 0.1922   & 0.8652   & -          \\
                                               & 0.005   & 0.1393 & -       & 0.1730   & 0.8473   & -          \\
                                               & 0.010   & 0.1620 & -       & 0.1814   & 0.8738   & -          \\
                                               & 0.050   & 0.3149 & -       & 0.1752   & 0.8624   & -          \\ \bottomrule
\end{tabular}
\caption{Comparison of simulation results under cases 1 and 2. For each estimator, there are 4 rows: each row represents one simulation result with 100 bootstrap samples. Only the unbalanced OT method uses a different penalty parameter. The differences among 4 rows in all other methods are due to the bootstrap randomness. }
\label{tb:simulation}
\end{table}

\noindent\textbf{Case 1.}

For the first simulation, we sample the treated and control groups from two different mixture normal distributions with 2 components. The treated and control groups only have a proportion of covariates in the overlapping area.
\begin{align*}
    &\mathbf{X}|\{T=0\}\sim 0.5\cdot N\left( -\mathbf{1},2I\right) + 
    0.5\cdot N\left( 0.5\mathbf{1},I\right) ,\\
    &\mathbf{X}|\{T=1\}\sim 0.5\cdot N\left( \mathbf{1},2I\right) + 
    0.5\cdot N\left( 0.5\mathbf{1},I\right)
\end{align*}


\noindent\textbf{Case 2.}
For the second simulation, the treated group is concentrated within the region of the control group.
\begin{align*}
    &\mathbf{X}|T=0\sim 0.5\cdot N\left( -\mathbf{1},2I\right) + 
    0.5\cdot N\left( 0.5\mathbf{1},I\right) ,\\
    &\mathbf{X}|T=1\sim 0.5\cdot N\left( \mathbf{1},0.5I\right) + 
    0.5\cdot N\left( 0.5\mathbf{1},0.5I\right)
\end{align*}

To compare the performance of unbalanced OT with KNN and IPW, 100 bootstrap samples are drawn in each case, the estimated ATE and ATT are computed for each method. For unbalanced OT, the simulation study includes a varying range of the penalty parameter $\epsilon$ from 0.001 to 0.05. In addition, we also compare the estimated standard deviation for ATT by unbalanced OT and KNN methods. Since IPW is a weighting-based method, whereas KNN and OT are imputation-based methods, the standard deviation for IPW is not included. The unadjusted ATE, defined as the observed sample difference between $Y(1)$ and $Y(0)$, is computed as a baseline estimation. Table \ref{tb:simulation} gives the summary results of different methods under case 1 and 2. For each case, to evaluate the estimation accuracy, we compute the mean of the absolute difference between the estimates and the true values across all bootstrap samples, denoted as \texttt{ATE\_diff}, \texttt{ATT\_diff}, and \texttt{ATT\_sd\_diff} respectively in Table \ref{tb:simulation}.

Based on the results in Table \ref{tb:simulation}, we have the following observations: 
\begin{enumerate}
    \item The smaller penalty term leads to better estimation accuracy for the unbalanced OT.
    \item OT is more biased when imputing potential outcome for outliers. Compared to case 1, everyone in the treated group is surrounded by controls in case 2. Hence the unbalanced OT has a better performance in estimating ATT in case 2. In contrast, in case 2 the ATE estimation for OT is worse than in case 1. This is because the potential outcome for the control group is biased due to lack of surrounding treated samples.
    \item OT possesses a lower standard deviation in general compared to KNN.
\end{enumerate}
\begin{figure}[hb!]
\centering
\includegraphics[width=\textwidth]{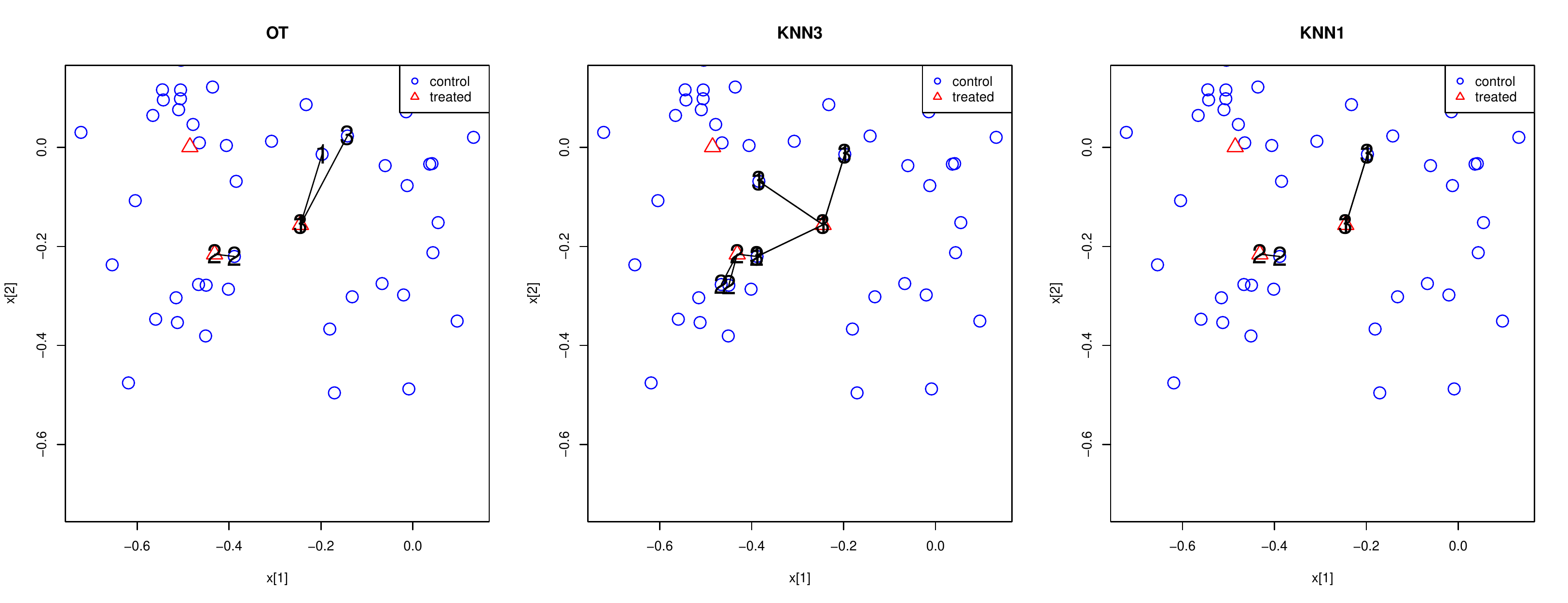}
\caption{Top 3 pairs with the largest unbalanced OT weights in the simulation case 2,  and the corresponding KNN matches for the selected treated individuals.}
\label{plot:compare_weight}
\end{figure}

To further illustrate the weight distribution in the unbalanced OT map, Figure \ref{plot:compare_weight} provides an enlarged covariates scatter plot, where the top 3 pairs that are assigned with the largest weights are marked. Different from KNN(k=1), OT assigns large weights from the treated ``1'' to both controls ``1'' and ``3''.

Overall, we recommend using unbalanced OT matching for estimating ATT when the overlap assumption is likely to be satisfied at least for one group, like in case 2, where the treated group is concentrated within the controls.

\subsection{Matching the Lalonde dataset}

We provide an application to the classical \texttt{Lalonde} dataset. The data are obtained from \url{https://users.nber.org/~rdehejia/data/.nswdata2.html} and are slightly larger and more comprehensive than the version used in \citet{dehejia1999causal} or \citet{imbens2015matching}. The dataset contains covariates and earning outcomes for men who were randomly selected into the Nationally Supported Work (NSW) program. The main outcome of interest is earnings in $1978$, \texttt{re78}. We match on 7 covariates: for continuous variables, \texttt{age} ranges from 17 to 55, \texttt{education} is the years of schooling, ranging from 3 to 16, \texttt{re75} indicates earnings in 1975, ranging from 0 to 37431.7; for the discrete variables, \texttt{black}, \texttt{hispanic} equal to 1 if an individual is Black or Hispanic, respectively,\texttt{married} is the indicator variable for martial status, \texttt{nodegree} equal to 1 if an individual is a high school dropout. We standardize \texttt{age}, \texttt{education} and \texttt{re75}. Among the 722 individuals in this dataset, 297 are in the treated group. Table \ref{tb:lalonde_sumstats} provides summary statistics of the covariates stratified by treatment groups.

\begin{table}[h]
\centering
\begin{tabular}{lrrrr}
\hline
                    & \multicolumn{2}{c}{\textbf{Treated (N=297)}}                                      & \multicolumn{2}{l}{\textbf{Control (N=425)}}                                      \\ \hline
\textbf{Continuous} & \textbf{mean}                           & \textbf{sd}                             & \textbf{mean}                           & \textbf{sd}                             \\
\texttt{age}                 & 24.63                                   & 6.69                                    & 24.45                                   & 6.59                                    \\
\texttt{education}           & 10.38                                   & 1.82                                    & 10.19                                   & 1.62                                    \\
\texttt{re75}                & 3066.10                                 & 4874.89                                 & 3026.68                                 & 5201.25                                 \\ 
\texttt{re78}                & 5976.35                              & 6923.80                                & 5090.05 & 5718.09                               \\\midrule 
\textbf{Discrete}   & \multicolumn{1}{l}{\textbf{Category 0}} & \multicolumn{1}{l}{\textbf{Category 1}} & \multicolumn{1}{l}{\textbf{Category 0}} & \multicolumn{1}{l}{\textbf{Category 1}} \\
\texttt{black}               & 59                                      & 238                                     & 85                                      & 340                                     \\
\texttt{hispanic}            & 269                                     & 28                                      & 377                                     & 48                                      \\
\texttt{married}             & 247                                     & 50                                      & 358                                     & 67                                      \\
\texttt{nodegree}            & 80                                      & 217                                     & 79                                      & 346                   \\
\bottomrule
\end{tabular}
\caption{Summary statistics of the Lalonde data.}
\label{tb:lalonde_sumstats}
\end{table}

Table \ref{tb:full_lalonde} provides the estimates for ATE and ATT under different methods. The penalty parameter is set to $10^{-3}$ for OT matching.
From Table \ref{tb:full_lalonde}, especially the ATT estimates, KNN(k=1) is almost half the value of OT and KNN(k=3) estimates. There are 2 possible explanations: (i) the noise in the outcome is large, in which case KNN(k=3) and OT results would be more trustworthy, since they are both based on weighted averages of nearby outcomes; (ii) the nearby controls of each treated individual are not very close in distance due to the lack of enough control observations, in which case KNN(k=1) estimate is more trustworthy since it only selects the nearest control as the match. We believe that the first explanation is more likely because the covariate distributions to be matched are mostly binary or discrete, which does not introduce a severe curse of dimensionality, despite the relatively small number of observations.

\begin{table}[h]
\centering
\begin{tabular}{lrrrrr}
\toprule
  & OT & IPW & KNN(k=3) & KNN(k=1) & Unadjusted\\
ATE & 760.7416 & 5686.583 & 722.4662 & 910.0533 & 886.3037\\
ATT & 828.3405 & 2381.305 & 870.6201 & 481.8134 & -\\
\bottomrule
\end{tabular}

\caption{ATE and ATT estimates for the full Lalonde data.}
\label{tb:full_lalonde}
\end{table}

The estimates of our method and the $3$-nearest neighbor method are qualitatively similar, even though we estimate a slightly higher average treatment effect of $\$760.74$ and a slightly lower average treatment effect on the treated group  of $\$828.34$. The sizable difference in outcome between the $3$-nearest and $1$-nearest neighbor estimator displays a challenge for a method that has a discrete penalty term: it is not quite clear at face-value which is the less biased estimate because there are no intermediate values, compared to the case of a continuous penalty term as in our method, which can better illustrate the sensitivity of the outome with respect to the penalty much more clearly. Our additional estimate implies that the less biased estimate among the $k$-nearest neighbor approaches is the $3$-nearest neighbor estimate.

\section{Conclusion}
We propose a matching approach for causal inference in observational studies based on unbalanced multimarginal optimal transportation. It can be understood as an alternative to the $k$-nearest neighbor method where the number of neighbors and optimal weights are chosen endogenously. The user instead specifies a penalty term $\varepsilon$. This proposed method possesses desirable properties: the optimal weights for the matches are automatically adjusted to the function that captures the cost of matching, which is usually chosen to be the squared Euclidean distance. Moreover, the method discards bad matches from either group. Together, these two properties lead to demonstrably competitive estimates in terms of bias and variance in finite samples compared to existing methods, such as the classical $k$-nearest neighbor- or propensity score methods. An application to the classical Lalonde dataset \citep{lalonde1986evaluating} shows that our method provides estimates that are qualitatively similar to those from the $3$-nearest neighbor approach.

In addition, we prove that, for fixed penalty term $\varepsilon$, the potential functions related to the optimal matching converge at the parametric rate to the true potential functions; moreover, the classical bootstrap method can be used to obtain confidence regions for the potential functions of the optimal matching. So, while a positive penalty term $\varepsilon$ induces some asymptotic bias in the estimator, this bias is still competitive with respect to those of existing methods, such as variants of the $k$-nearest neighbor matching approach. 

The method also possesses good computational properties with a strong foundation in computational optimal transport theory \citep{peyre2019computational}. This allows the applied researcher to leverage the latest advances in these areas for matching in high-dimensional spaces. A simple simulation exercise demonstrates the method's competitive finite sample properties in comparison to classical $k$-nearest neighbor approaches. Finally, the proposed method allows researchers to estimate a joint matching over several treatment arms, which could be interesting in answering general non-linear causal questions in this setting. 

From a mathematical perspective, we provide a statistical analysis of the potential functions of the optimal matching in the generalized Sinkhorn algorithm for a fixed penalty term $\varepsilon$. This complements existing statistical and mathematical results for the value function of the unbalancd optimal transport problem \citep{sejourne2019sinkhorn}, as well as results for the classical balanced entropy regularized optimal transport problem \citep{carlier2021linear, carlier2020differential, di2020optimal, eckstein2021quantitative, nutz2021entropic,  genevay2019sample, klatt2020empirical, pooladian2021entropic}, especially that in \citet{harchaoui2020asymptotics}. A next step in this direction is a finer analysis of the variance functional of the (potential functions of the) optimal matching in the unbalanced optimal transport case, also in the setting where the penalty term is allowed to vanish, similar to the contributions in the classical balanced setting. Finally, another interesting question is to analyze which matches are chosen by the method and which observations are discarded as a function of the penalty term. This could yield insights to optimize the matching approach even further and position unbalanced optimal transportation to potentially be used as a general method for nonparametric regression approaches.

\bibliography{main_bibfile}
\appendix
\section{Proofs}
\subsection{Proof of Proposition \ref{prop:consistency}}
\begin{proof}
The proof is analogous to the results in \citet{sejourne2019sinkhorn}, but for the multimarginal case. Since $c$ is smooth, and in particular Lipschitz, by Assumption \ref{ass:cost}, the same proof as in Lemma 7 of \citet{sejourne2019sinkhorn} implies that the dual problem \eqref{eq:dual} can be restricted to a compact subset of $\prod_{j=1}^J C(\mathcal{X}_j)$, the product of spaces of continuous functions. A straightforward generalization of the proof of Proposition 8 in \citet{sejourne2019sinkhorn} to the multimarginal case implies that a set $\{\varphi_j\}_{j=1}^J$ of potentials is optimal if and only if it is a fixed point of the generalized Sinkhorn matching \eqref{eq:FOC}. Then the same proof as in Theorem 1 of \citet{sejourne2019sinkhorn} implies that the generalized Sinkhorn algorithm \eqref{eq:sinkhornconv} converges to an optimal solution. This proves the first part of the statement. 

The consistency result for the optimal potentials follows from a similar argument. Again, under Assumptions \ref{ass:div} and \ref{ass:cost}, a straightforward generalization of Proposition 10 in \citet{sejourne2019sinkhorn} implies that as measures $\hat{\mu}_{N_j}$ each converge weakly to their population counterparts $\mu_j$, the optimal potentials $\hat{\varphi}_{N_j}^*$ converge uniformly to their population counterparts $\varphi_j^*$. Now, the measures $\hat{\mu}_{N_j}$ are actually random measures, and the classical result in \citet{varadarajan1958convergence} implies that they converge weakly almost surely to their population counterparts and hence weakly in probability. This implies directly that the corresponding optimal potentials converge uniformly in probability. 

The weak convergence in probability of the couplings $\hat{\gamma}^*_N$ to the optimal coupling $\gamma^*$ now follows from the extended continuous mapping theorem \citep[Theorem 1.11.1 (iii)]{wellner2013weak} and the fact that the exponential in \eqref{eq:optim_measure} is continuous with respect to the optimal potentials in combination with Slutsky's lemma \citep[Example 1.4.7]{wellner2013weak}.
\end{proof}

\subsection{Proof of Theorem \ref{thm:potential}}
For the sake of notation we drop the iteration $m$ in the definition of the estimated potentials, i.e.~we write $\hat{\varphi}_N$ instead of $\hat{\varphi}_N^m$, unless we need to specify $m$ explicitly. 

We split the proof into another lemma and the main proof. The additional lemma provides the consistency of the empirical analogues of the potential functions $\hat{\varphi}_N$ in the space $\prod_jC^\infty(\mathcal{X}_j)$. We say that $f_j\in C^\infty(\mathcal{X}_j)$ if $f_j\in C^k(\mathcal{X}_j)$ for every multi-index $k\coloneqq(k_1,\ldots,k_d)\in\mathbb{N}^d_0$, i.e.~if 
\[ \left\|f_j^{(k)}\right\|_{\infty} \coloneqq \sup_{x_j\in\mathcal{X}_j}\left\lvert f_j^{(k)}(x_j)\right\rvert<+\infty\quad\text{for all $k\in\mathbb{N}^d_0$}\] where the expression $f_j^{(k)}$ is shorthand for the $k$-th partial derivative of $f_j$:
\[f_j^{(k)}(x_j) = \frac{\partial^{|k|}}{\partial^{k_1}x_{j1}\cdots\partial^{k_d}x_{1d}}f_j(x_{j1},\ldots,x_{jd})\] and $|k| = \sum_{i=1}^d k_i$. $\mathbb{N}_0$ denotes the natural numbers including zero. We write 
\[\|\varphi\|_{C^k}\coloneqq \sum_{j=1}^J \sum_{0\leq k\leq |k|} \|\varphi_j^{(k)}\|_{\infty}\coloneqq \sum_{j=1}^J \sum_{0\leq k\leq |k|} \sup_{x_j\in\mathcal{X}_j}\left\lvert \varphi_j^{(k)}(x_j)\right\rvert\] for the $C^k$-norm of $\varphi\equiv (\varphi_1,\ldots, \varphi_J)$.

The following lemma provides the consistency of the optimal potentials $\hat{\varphi}_N$ and all their derivatives as the empirical measures $\bigotimes_j\hat{\mu}_{N_j}$ converge weakly to $\bigotimes_j\mu_j$. 
\begin{lemma}\label{lem:consistency}
Under Assumptions \ref{ass:iid} - \ref{ass:cost}
\[\left\|\hat{\varphi}_N - \varphi^*\right\|_{C^k} = o_P(1)\]
for any multi-index $k\in\mathbb{N}_0^d$.
\end{lemma}
\begin{proof}
We start with the case $|k|=0$. By Proposition \ref{prop:consistency}
\[\left\|\hat{\varphi}_N-\varphi^*\right\|_\infty = o_P(1).\] 

Based on this, the convergence of the derivatives now follows from the Arzel\`a-Ascoli theorem. We show this by an induction over the multi-indices $k\in\mathbb{N}_0^d$. By a multivariate extension of the classical Fa\`a di Bruno formula \citep[Corollary 2.10]{constantine1996multivariate}, one can show \citep[equation (82)]{sejourne2019sinkhorn} that all derivatives of the potentials are uniformly bounded for fixed $\varepsilon$, i.e.
\[\left\|\left(\varphi^*\right)^{(k)}\right\|_\infty \leq M<+\infty,\] where $M$ is a constant that depends on $\varepsilon$, the dimension $d$, the number of elements $J$, the size of the supports $\mathcal{X}_j$, and the bound $\|c^{(k)}\|_{\infty}$ on all derivatives of the cost function.

Now consider the sequences $\left\{\hat{\varphi}_{N}\right\}^{(k)}_{N\in\mathbb{N}}$ with multi-indices with $|k|=1$. By the argument above, these sequences are uniformly bounded and equicontinuous. Uniform boundedness follows from the fact that $M$ is independent of the measures in question. Equicontinuity follows from the fact that all derivatives with multi-indices such that $|k|=2$ are uniformly bounded and the mean value inequality. 

But then by the Arzel\`a-Ascoli theorem, the sequence $\{\hat{\varphi}^{(k)}_{N}\}$ with $|k|=1$ has a subsequence $\{\hat{\varphi}^{(k)}_{N_m}\}$ that converges uniformly. Since $\{\hat{\varphi}_{N}\}$ converges uniformly to $\varphi^*$, this subsequence converges uniformly to $\left(\varphi^*\right)^{(k)}$. Now since the original sequence $\{\hat{\varphi}_{N}\}$ converges uniformly, this implies that any subsequence converges uniformly to the same element. So repeating the same argument as above for any subsequence $\{\hat{\varphi}_{N_m}\}$ shows that any such subsequence has a further subsequence $\{\hat{\varphi}_{{N_m}_l}\}$ such that $\{\hat{\varphi}^{(k)}_{{N_m}_l}\}$ converges uniformly to $\left(\varphi^*\right)^{(k)}$ for $|k|=1$. But this implies that any subsequence $\{\hat{\varphi}^{(k)}_{{N_m}_l}\}$ converges uniformly to $\left(\varphi^*\right)^{(k)}$ for the multi-index with $|k|=1$. This implies that the sequence $\{\hat{\varphi}^{(k)}_{N}\}$ converges uniformly to $\left(\varphi^*\right)^{(k)}$ for $|k|=1$.

Now proceed by induction to finish the proof. We have just shown that the sequence $\{\hat{\varphi}^{(k)}_{N}\}$ converges to $\left(\varphi^*\right)^{(k)}$ for $|k|=1$. We know by the above argument that all $\hat{\varphi}^{(k)}_{N}$ with $|k|=2$ and $|k|=3$ are uniformly bounded which also implies that the sequence $\hat{\varphi}^{(k)}_{N}$ with $|k|=2$ is equicontinuous. Apply the same reasoning as above. Doing this for all multi-indices $k\in\mathbb{N}^d_0$ finishes the proof.
\end{proof}

We can now prove Theorem \ref{thm:potential}.
\begin{proof}[Proof of Theorem \ref{thm:potential}]
We prove the statement using standard results from empirical process theory \citep{wellner2013weak, kosorok2008introduction}, in particular results for $Z$-estimators. We consider the integrated version of the Schr\"odinger system \eqref{eq:FOC} a $Z$-estimator in the sense of Theorem 3.3.1 in \citet{wellner2013weak} and Theorem 13.4 in \citet{kosorok2008introduction}. In particular, the first-order equations \eqref{eq:FOC} are derived by requiring that the first variation $\delta F(\varphi^*,v)$ \citep[Definition 40.2]{zeidler2013nonlinear} of the objective function \eqref{eq:dual} in any direction $v\in \prod_{j=1}^JL^1(\mu_j)$ vanishes at the optimal $\varphi^*$, i.e.
\begin{equation}\label{eq:firstvariation}
\begin{aligned}
0=\delta F(\varphi^*,v)\coloneqq&\sum_{j=1}^J \int_{\mathcal{X}} v_j\left[{\phi^*}'\left(-\varphi_j\right)-\exp\left(\frac{\sum_{j=1}^J\varphi_j-c}{\varepsilon}\right)\right]\dd\bigotimes_{j=1}^J\mu_j\\
& = \sum_{j=1}^J\int_{\mathcal{X}_j} v_j\int_{\mathcal{X}_{-j}}\left[{\phi^*}'\left(-\varphi_j\right)-\exp\left(\frac{\sum_{j=1}^J\varphi_j-c}{\varepsilon}\right)\right]\dd\bigotimes_{i\neq j}\mu_i\dd\mu_j,\\
\end{aligned}
\end{equation}
where ${\phi^*}'$ is the first derivative of $\phi^*$. We consider  $v\in \prod_j L^1(\mu_j)$ because $\sum_jv_j$ and the term in brackets in \eqref{eq:firstvariation} form a natural duality bracket in $\left(\prod_jL^1(\mu_j),\prod_jL^\infty(\mu_j)\right)$. 

Now we show that \eqref{eq:firstvariation} holds if and only if the first order conditions \eqref{eq:FOC} are satisfied. Since all measures $\mu_j$ are probability measures by assumption in our setting, Euler-Lagrange equations corresponding to \eqref{eq:firstvariation} are 
\[{\phi^*}'\left(-\varphi_j\right)=\int_{\mathcal{X}_{-j}}\exp\left(\frac{\sum_{j=1}^J\varphi_j-c}{\varepsilon}\right)\dd\bigotimes_{i\neq j}\mu_i\qquad \text{$\mu_j$-almost everywhere for all $j$.}\] 
By equation (3.6) in \citet{combettes2013moreau} this is equivalent to 
\[\varphi_j = \text{aprox}_{\phi^*}^\varepsilon\left(-\text{Smin}_j^\varepsilon\left(c-\sum_{i\neq j}\varphi_i\right)\right)\qquad\text{$\mu_j$-almost everywhere for all $j$},\] where
\[\text{Smin}_j^\varepsilon\left(c-\sum_{i\neq j}\varphi_i\right) \coloneqq -\varepsilon\log\int_{\mathcal{X}_{-j}}\exp\left(\frac{\sum_{i\neq j}\varphi_i-c}{\varepsilon}\right)\dd\bigotimes_{i\neq j}\mu_i\] is the ``soft-min'' operator.

We can now focus on proving the theorem. It will be based on Theorem 3.3.1 in \citet{wellner2013weak} and Theorem 13.4 in \citet{kosorok2008introduction}. We need to show the following:
\begin{itemize}
    \item[(i)] $\left\|\hat{\varphi}_N-\varphi^*\right\|_{\infty} = o_P(1)$
    \item[(ii)] There exists some $\delta>0$ such that the classes
    \begin{multline*}\tilde{\mathcal{F}}_{\delta,j}\coloneqq\left\{v_j\left[\varphi_j - \text{aprox}_{\phi^*}^\varepsilon\left(-\text{Smin}_j^\varepsilon\left(c-\sum_{i\neq j}\varphi_i\right)\right)\right]: \|\varphi - \varphi^*\|_{C^k}<\delta\thickspace \forall k\in\mathbb{N}^d_0,\right.\\ \left. \thickspace v_j\in B_1\left(C^\infty_c(\mathcal{X}_j)\right), \varphi\in \prod_jC^\infty(\mathcal{X}_j)\right\}
    \end{multline*} are $\mu_j$-Donsker for all $j$ (see chapters 2.1 and 2.5 in \citet{wellner2013weak} for a definition of Donsker classes). Here, $C^\infty_c(\mathcal{X}_j)$ is the space of all smooth functions with compact support on $\mathcal{X}_j$ and $B_1\left( C^\infty_c(\mathcal{X}_j)\right)$ is the unit ball in $C_c^\infty(\mathcal{X}_j)$, i.e.~the set of all functions $v_j$ with $\|v_j\|_{C^k(\mathcal{X}_j)}\leq 1$ for any multi-index $k\in\mathbb{N}^d_0$.
    \item[(iii)] As $\|\varphi-\varphi^*\|_{\infty}\to0$
    \[\sup_{v_j\in B(C_c^\infty(\mathcal{X}_j))}\int_{\mathcal{X}_j}v_j^2\left[T_j(\varphi)(x_j) - T_j(\varphi^*)(x_j)\right]^2\dd\bigotimes_{ j}\mu_j\to0,\quad\text{$j=1,\ldots, J$}\]
    \item[(iv)] $T(\varphi)\coloneqq (T_1(\varphi),\ldots, T_J(\varphi))$ is Fr\'echet-differentiable with continuously invertible derivative 
    \[T'_j(\varphi)(u) = u_j+\varepsilon Y\cdot\frac{\int_{-\mathcal{X}}\sum_{i\neq j}u_i\exp\left(\frac{\sum_{i\neq j}\varphi_i-c}{\varepsilon}\right)\dd\bigotimes_{i\neq j}\mu_i}{\int_{-\mathcal{X}}\exp\left(\frac{\sum_{i\neq j}\varphi_i-c}{\varepsilon}\right)\dd\bigotimes_{i\neq j}\mu_i}\] for all $u\coloneqq (u_1,\ldots,u_J)\in\prod_{j=1}^J C(\mathcal{X}_j)$, where
    \[Y \equiv \frac{{\phi^*}'\left[\text{aprox}_{\phi^*}^\varepsilon\left(-\text{Smin}_j^\varepsilon\left(c-\sum_{i\neq j}\varphi_i\right)\right)\right]}{{\phi^*}'\left[\text{aprox}_{\phi^*}^\varepsilon\left(-\text{Smin}_j^\varepsilon\left(c-\sum_{i\neq j}\varphi_i\right)\right)\right]+ \varepsilon {\phi^*}''\left[\text{aprox}_{\phi^*}^\varepsilon\left(-\text{Smin}_j^\varepsilon\left(c-\sum_{i\neq j}\varphi_i\right)\right)\right]},\] with ${\phi^*}''$ being the second derivative of $\phi^*$.

    \item[(v)] $\left\|\hat{T}_N(\hat{\varphi}_N)\right\|_{\infty} = o_P(N^{-1/2})$ for $\hat{\varphi}_N$ and 
    \[P\left( \left\|\left(\sqrt{N_1}\hat{T}_{N_1}^b(\hat{\varphi}^b_{N}), \ldots, \sqrt{N}_J \hat{T}_{N_J}^b(\hat{\varphi}^b_{N})\right)\right\|_{L^\infty}>\eta\rvert \mathbb{X}_N\right)=o_P(1),\]
    where $\mathbb{X}_N$ denotes the sample drawn and where $\hat{T}^b_N(\hat{\varphi}^b_N)$ denotes a bootstrapped version of the operator \citep[Theorem 13.4]{kosorok2008introduction}, i.e.~where we estimate the measures $\mu_j$ via some (multiplier) bootstrap procedure, see \citet[section 3.6]{wellner2013weak} for details on the bootstrap. 
\end{itemize}
We prove each requirement one by one before putting everything together.\\

\noindent\emph{Requirement (i):}
This follows from Lemma \ref{lem:consistency}. \\

\noindent\emph{Requirement (ii):} 
This requirement is needed to provide the asymptotic equicontinuity of the empirical process map corresponding to the estimation equation 
\[\delta F(\varphi^*,v) =0\quad\text{for all $v\in  \prod_jB_1\left(C^\infty_c(\mathcal{X}_j)\right)$},\] where $\prod_jB_1\left(C^\infty_c(\mathcal{X}_j)\right)$ denotes the unit ball of all smooth functions $v\equiv (v_1,\ldots, v_J)$ with compact support on $\prod_j\mathcal{X}_j$. The space of smooth functions is a Fr\'echet space, so a general way to consider the unit ball in this Fr\'echet space is to define a weighted series of the $C^k$-norms that makes the ``$C^\infty$-norm'' summable, see \citet{treves2006topological} for details. The weak convergence of the empirical process of the first-order condition required for obtaining the asymptotic distribution of the empirical process of the potentials can then be shown by the weak convergence of the empirical process corresponding to the first variation \eqref{eq:firstvariation} for all $v$ in a sufficient index set $\mathcal{S}\subset \prod_{j=1}^JL^1(\mu_j)$ by Lemma 3.3.5 in \citet{wellner2013weak}. A sufficient condition that also provides the fact that this asymptotic distribution will be a tight mean-zero Gaussian process is exactly requirement (ii) above for $\mathcal{S}= \prod_j B_1\left(C^\infty_c(\mathcal{X}_j)\right)$ and follows from Theorem 13.4 (C) in \citet{kosorok2008introduction}. 

So to show requirement (ii) we hence have to show two things. First, we need to show that we can set $\mathcal{S}\coloneqq \prod_jB_1\left(C^\infty_c(\mathcal{X}_j)\right)$. Second, we need to show that $\mathcal{F}_{\delta,j}$ are indeed $\mu_j$-Donsker. For the first statement note that
$T$ is Fr\'echet differentiable with derivative $T_j'(\varphi)$ for all $\varphi$, as the G\^ateaux- derivative obtained via the first variation \eqref{eq:firstvariation} is continuous in a neighborhood of $\varphi$ and hence coincides with the Fr\'echet-derivative \citep[p.~192]{zeidler2013nonlinear}. Now Fr\'echet differentiability in particular implies that the convergence of the linear approximation to the derivative at $\varphi^*$ from any direction $v\in\prod_jL^1(\mu_j)$ is of the same rate. This further implies that we can infer the derivative in any direction $v\in \prod_j L^1(\mu_j)$ by only considering directions in a dense subset of $\prod_j L^1(\mu_j)$. In particular, the equality \eqref{eq:firstvariation} only has to hold for all $v$ in a dense subset of the unit ball of $\prod_{j=1}^JL^1(\mu_j)$. But $\prod_{j=1}^J C_c^\infty(\mathcal{X}_j)$ is dense in $\prod_{j=1}^J L^1(\mu_j)$ \citep[Proposition 8.17]{folland1999real}, so that we can focus on the unit ball in $\prod_j C_c^\infty(\mathcal{X}_j)$.  

For the second criterion, we need to show that there exists some $\delta>0$ such that $\mathcal{F}_{\delta,j}$ is $\mu_j$-Donsker for all $j$. For this, we use the concept of bracketing entropy and bracketing numbers \citep[chapter 2.7]{wellner2013weak}. Since $\prod_jC(\mathcal{X}_j)$ is a Banach lattice under the standard ordering $f\preceq g \Leftrightarrow f_j(x_j)\leq g_j(x_j)$ for all $j$, we can define, for any two functions $f_l, f_u\in \prod_j C(\mathcal{X}_j)$ the bracket $[f_l,f_u]$ as the set of all functions $f$ such that $f_l\preceq f\preceq f_u$. A $\eta$-bracket is a bracket $[f_l,f_u]$ with $\|f_u-f_l\|_{L^\infty}<\eta$. For some set $\mathcal{F}$ of functions the bracketing number $N_{[\medspace]}\left(\eta,\mathcal{F},\|\cdot\|\right)$ is the minimum number of $\eta$-brackets needed to cover $\mathcal{F}$ in the metric induced by $\|\cdot\|$ \citep[p.~83]{wellner2013weak}. 

To show that $\mathcal{F}_{\delta,j}$ is $\mu_j$-Donsker for some $\delta>0$ we want to show that the corresponding bracketing entropy integral is finite, i.e.~that
\begin{equation}\label{eq:entropyint}
\int_0^\infty \sqrt{\ln N_{[\thinspace]}\left(\eta,\mathcal{F}_{\delta,j}, L^2(\mu_j)\right)}d\eta<+\infty.
\end{equation}
The idea is to treat $v_j$ and 
\[g_\varphi(x_1,\ldots,x_J) \equiv g_\varphi(x) \coloneqq \varphi_j - \text{aprox}_{\phi^*}^\varepsilon\left(-\text{Smin}_j^\varepsilon\left(c-\sum_{i\neq j}\varphi_i\right)\right)\] separately.
The reason is that Corollary 2.7.2 in \citet{wellner2013weak} in combination with the fact that $\mathcal{X}_j\subset\mathbb{R}^d$ is compact implies that $B_1(C_c^\infty(\mathcal{X}_j))$ is $\mu_j$-Donsker for any dimension $d$. Moreover, it is a uniformly bounded set. If we can show that the set of all functions $g_\varphi$ for some $\delta>0$ is also a uniformly bounded Donsker class for any dimension $d$, then their product, i.e.~$\mathcal{F}_{\delta,j}$, will be a Donsker class for any dimension by Example 2.10.8 in \citet{wellner2013weak}. 

Using the multivariate Fa\`a di Bruno formula \citep{constantine1996multivariate}, the optimal $\varphi^*$ obtained via the Sinkhorn algorithm is bounded and smooth \citep[equation (82)]{sejourne2019sinkhorn}. Since this holds for any set of measures $\{\mu_j\}_{1\leq j\leq J}$ it also holds for $\hat{\varphi}_n$. Moreover, by Lemma \ref{lem:consistency} we know that $\hat{\varphi}_{N_j}$ converges to $\varphi^*_j$ in $C^\infty(\mathcal{X}_j)$. Hence for some small $\delta>0$ we can find a uniform bound for the $\hat{\varphi}_{N_j}$ in $C^\infty(\mathcal{X}_j)$ by the reverse triangle inequality:
\[\left\lvert \left\|\hat{\varphi}_N\right\|_{C^k} - \left\|\varphi^*\right\|_{C^k}\right\rvert\leq \left\|\hat{\varphi}_N-\varphi^*\right\|_{C^k}<\delta,\] for some $\delta>0$ and all multi-indices $k\in\mathbb{N}_0^d$. This implies that the classes
\[\mathcal{G}_{\delta,j}\coloneqq\left\{\left[\varphi_j - \text{aprox}_{\phi^*}^\varepsilon\left(-\text{Smin}_j^\varepsilon\left(c-\sum_{i\neq j}\varphi_i\right)\right)\right]: \|\varphi - \varphi^*\|_{C^k}<\delta\thickspace\forall k\in\mathbb{N}_0^d, \varphi\in \prod_jC^\infty(\mathcal{X}_j)\right\}\] are a bounded subset of $\prod_j C^\infty(\mathcal{X}_j)$ for compact and convex $\mathcal{X}_j$. Therefore by Corollary 2.7.2 $\mathcal{G}$ is a uniformly bounded Donsker class for any dimension $d$. This implies by Example 2.10.8 in \citet{wellner2013weak} that $\tilde{\mathcal{F}}_{\delta,j}$ is a $\bigotimes_j\mu_j$-Donsker class in any dimension $d$.\\

\noindent\emph{Requirement (iii):} 
This follows immediately by the fact that the operators $T_j$ are smooth and H\"older's inequality.\\

\noindent\emph{Requirement (iv):}
We want to show that the linear integral operator $T'(\varphi)$ is invertible on $\prod_j C(\mathcal{X}_j)$ with bounded inverse. To see this, first write it as 
\[T'(\varphi)(u) = \text{Id} - L(u),\] where $\text{Id}$ is the identity operator and $L(u): \prod_j C(\mathcal{X}_j)\to \prod_j C(\mathcal{X}_j)$ with 
\[L_j(u)\coloneqq -\varepsilon Y\cdot\frac{\int_{-\mathcal{X}}\sum_{i\neq j}u_i\exp\left(\frac{\sum_{i\neq j}\varphi_i-c}{\varepsilon}\right)\dd\bigotimes_{i\neq j}\mu_i}{\int_{-\mathcal{X}}\exp\left(\frac{\sum_{i\neq j}\varphi_i-c}{\varepsilon}\right)\dd\bigotimes_{i\neq j}\mu_i}\] is a compact operator on $\prod_j C(\mathcal{X}_j)$, because it has a smooth kernel defined on compact subsets which makes the kernel uniformly bounded; therefore, Exercise 3.13 in \citet{abramovich2002invitation} implies that $L_j(u)$ is a compact operator on $C(\mathcal{X}_j)$, which implies that $L(u)$ is a compact operator on the (finite) product space. We now show that $T'$ is invertible on $\prod_j C(\mathcal{X}_j)$. To do this, we bound the operator norm of $L$ on this product space, which we define by \citep[Exercise 3.13]{abramovich2002invitation}
\[\left\|L(u)\right\|_{op} = \max_j\|L_j(u)\|_{op} = \max_j\sup_{\|u\|_{\infty}} \|L_j(u)\|_{\infty}.\] Note that this norm is equivalent to the norm
$\left\|L(u)\right\|_{op} = \sum_j\|L_j(u)\|$, since we have a finite product, but we can show invertibility more easily in the former.
But since $\exp(\cdot)$ is non-negative, we have by H\"older's inequality \citep[Theorem 6.8]{folland1999real}
\begin{align*}
   \left\lvert L_j(u)\right\rvert =&  \left\lvert\varepsilon Y\cdot\frac{\int_{-\mathcal{X}}\sum_{i\neq j}u_i\exp\left(\frac{\sum_{i\neq j}\varphi_i-c}{\varepsilon}\right)\dd\bigotimes_{i\neq j}\mu_i}{\int_{-\mathcal{X}}\exp\left(\frac{\sum_{i\neq j}\varphi_i-c}{\varepsilon}\right)\dd\bigotimes_{i\neq j}\mu_i}\right\rvert\\
   \leq & \left\lvert\varepsilon Y\cdot\frac{\left\|\sum_{i\neq j}u_i\right\|_{L^\infty(\bigotimes_{i\neq j}\mu_i)} \left\|\exp\left(\frac{\sum_{i\neq j}\varphi_i-c}{\varepsilon}\right)\right\|_{L^1(\bigotimes_{i\neq j}\mu_i)}}{\int_{-\mathcal{X}}\exp\left(\frac{\sum_{i\neq j}\varphi_i-c}{\varepsilon}\right)\dd\bigotimes_{i\neq j}\mu_i}\right\rvert\\
   = &\left\lvert\varepsilon Y\right\rvert \frac{\left\|\sum_{i\neq j}u_i\right\|_{L^\infty(\bigotimes_{i\neq j}\mu_i)} \left\|\exp\left(\frac{\sum_{i\neq j}\varphi_i-c}{\varepsilon}\right)\right\|_{L^1(\bigotimes_{i\neq j}\mu_i)}}{\int_{-\mathcal{X}}\left\lvert\exp\left(\frac{\sum_{i\neq j}\varphi_i-c}{\varepsilon}\right)\right\rvert\dd\bigotimes_{i\neq j}\mu_i}\\
   =& \varepsilon Y \frac{\left\|\sum_{i\neq j}u_i\right\|_{L^\infty(\bigotimes_{i\neq j}\mu_i)} \left\|\exp\left(\frac{\sum_{i\neq j}\varphi_i-c}{\varepsilon}\right)\right\|_{L^1(\bigotimes_{i\neq j}\mu_i)}}{\left\|\exp\left(\frac{\sum_{i\neq j}\varphi_i-c}{\varepsilon}\right)\right\|_{L^1(\bigotimes_{i\neq j}\mu_i)}}\\
   =&\varepsilon Y \left\|\sum_{i\neq j}u_i\right\|_{L^\infty(\bigotimes_{i\neq j}\mu_j)}\\
   \leq &\varepsilon Y \left\|\sum_{i\neq j}u_i\right\|_{\infty}\\
   \leq &\varepsilon Y \sum_{i\neq j}\left\|u_i\right\|_{\infty}\\
   \leq &\varepsilon Y \sum_{j}\left\|u_j\right\|_{\infty},
\end{align*}
where $L^\infty(\mu_j)$ is the standard essential supremum norm \citep[chapter 6]{folland1999real} and because $Y$ is non-negative because the anisotropic proximity operator is non-decreasing \citep[Proposition 5]{sejourne2019sinkhorn}. The third-to-last line follows from the definition of the essential supremum. But for the operator norm we pick $u$ such that $\|u\|_\infty\equiv \sum_j\left\|u_i\right\|_{\infty}=1$, which implies that
\[\left\|L(u)\right\|_{op}\leq \varepsilon Y.\]
But by Assumption \ref{ass:div} the function $\phi$ is strictly convex and smooth, which implies that $\phi^*$ is strictly convex and smooth \citep[Theorem 26.6]{rockafellar1997convex}, which implies that ${\phi^*}''>0$. This implies that
\[Y \equiv \frac{{\phi^*}'\left[\text{aprox}_{\phi^*}^\varepsilon\left(-\text{Smin}_j^\varepsilon\left(c-\sum_{i\neq j}\varphi_i\right)\right)\right]}{{\phi^*}'\left[\text{aprox}_{\phi^*}^\varepsilon\left(-\text{Smin}_j^\varepsilon\left(c-\sum_{i\neq j}\varphi_i\right)\right)\right]+ \varepsilon {\phi^*}''\left[\text{aprox}_{\phi^*}^\varepsilon\left(-\text{Smin}_j^\varepsilon\left(c-\sum_{i\neq j}\varphi_i\right)\right)\right]}<1\]
for all $\varepsilon>0$, which in turn implies
$\|L(u)\|_{op}<1$ for $0<\varepsilon<1$. A standard Neumann series argument then implies that $\text{Id}-L$ is invertible for $0<\varepsilon<1$. The fact that $I-L$ maps between Banach spaces and is bounded and surjective by the Fredholm alternative \citep[Theorem 6.6]{brezis2011functional} in conjunction with the fact that its kernel is trivial because we have a unique solution to the dual problem, the open mapping theorem \citep[Theorem 5.9]{folland1999real} implies that $(I-L)^{-1}$ is bounded. \\

\noindent\emph{Requirement (v):}
This follows from the convergence of the generalized Sinkhorn iterations from Lemma \ref{lem:consistency} and the compactness of the operators for given measures $\mu_j$, which also holds when we replace $\mu_j$ with $\hat{\mu}_j$. \\

We can now put everything together by using Theorem 3.3.1 in \citet{wellner2013weak} and Theorem 13.4 in \citet{kosorok2008introduction}. We show that the five requirements above imply the requirements (A) to (F) in Theorem 13.4 in \citet{kosorok2008introduction}. (A) is implied by the fact that the optimum $\varphi^*$ is well-separated, which in turn follows from the fat that $T(\varphi)$ is strictly convex in $\varphi$ by Assumption \ref{ass:div}. (B) is implied by (i). (C) is implied by (ii). (D) is implied by (iii). (E) is implied by (v). Finally, (F) is implied by (iv). Hence Theorem 13.4 in \citet{kosorok2008introduction} implies that
\[\left(\sqrt{N_1}\left(\hat{\varphi}_{N_1} - \varphi_1^*\right),\ldots, \sqrt{N_J}\left(\hat{\varphi}_{N_J}-\varphi_J^*\right)\right) \rightsquigarrow - \left(T'(\varphi^*)\right)^{-1}(Z),\] 
where $Z$ is the tight mean-zero Gaussian limit process of 
\[\left(\sqrt{N_1} \left(\hat{T}_{N_1}(\varphi^*) - T_1(\varphi^*)\right),\ldots, \sqrt{N_J} \left(\hat{T}_{N_J}(\varphi^*) - T_J(\varphi^*)\right)\right)\]
in the space $\ell^\infty\left(\prod_j C(\mathcal{X}_j)\right)$ of bounded functions on $\prod_jC(\mathcal{X}_j)$. Moreover, Theorem 13.4 in \citet{kosorok2008introduction} directly implies that a standard bootstrap procedure is valid.
\end{proof}

\subsection{Proof of Corollary \ref{corr:density}}

\begin{proof}[Proof of Corollary \ref{corr:density}]
The proof relies on the classical delta method \citep[chapter 3.9]{wellner2013weak}. 

By the result of Theorem \ref{thm:potential}, 
\[\left(\sqrt{N_1}\left(\hat{\varphi}^m_{N_1} - \varphi_1^*\right),\ldots, \sqrt{N_J}\left(\hat{\varphi}^m_{N_J}-\varphi_J^*\right)\right) \rightsquigarrow - \left(T'(\varphi^*)\right)^{-1}(Z),\] 

We denote 
\[Z_\varphi\equiv(Z_{\varphi,1},\ldots,Z_{\varphi,J})\coloneqq -\left(T'(\varphi^*)\right)^{-1}(Z)\] as the resulting process. First we show that $Z_\varphi$ is a mean-zero tight Gaussian process.

By Theorem \ref{thm:potential} $T'(\varphi^*)(u)$ is an isomorphism from $\prod_jC(\mathcal{X}_j)$ to $\prod_jC(\mathcal{X}_j)$. By definition, $T'_j(\varphi^*)(h)$ is a continuous linear mapping w.r.t $h$ and $T'_j(\varphi^*)(0)=0$. By the linearity of $(T'(\varphi^*))^{-1}(h)$ and Fubini's theorem, $\mathbb{E} Z_\varphi = \mathbb{E} [ -\left(T'(\varphi^*)\right)^{-1}(Z)] = -\left(T'(\varphi^*)\right)^{-1}(0)=0$. By Proposition 7.5 in \cite{kosorok2008introduction}, $Z_\varphi$ is a mean-zero Gaussian process. In addition, because $\left(T'(\varphi^*)\right)^{-1}(h)$ is a bijection, $Z$ being tight implies that $Z_\varphi$ is a tight process. Therefore, $Z_\varphi$ is a mean-zero tight Gaussian process.

Define $\Phi:\prod_{j=1}^J C(\mathcal{X}_j) \to \prod_{j=1}^J C(\mathcal{X}_j)$ as
\[\Phi(\varphi_1,...,\varphi_J):=\exp\left( \frac{\sum_{j=1}^J\varphi_j(x_j) - c(x)}{\varepsilon} \right).\]
Then for any function $h\coloneqq (h_1,\ldots, h_J)\in \prod_{j=1}^J C(\mathcal{X}_j)$, and any sequences $\prod_{j=1}^J C(\mathcal{X}_j)\ni h_n\to h$ and $t_n\to 0$,
\begin{align*}
    \frac{\Phi(\varphi+t_n h_n) - \Phi(\varphi)}{t_n}&=\frac{\exp\left(\frac{1}{\varepsilon}\left( \sum_j\varphi_j+t_n h_{j,n} \right)\right) - \exp\left(\frac{1}{\varepsilon}\sum_j \varphi_j\right)}{t_n}
    \exp\left( -\frac{1}{\varepsilon}c \right)\\
    &=\frac{1}{t_n} \left( \frac{\sum_{j=1}^Jt_n h_{j,n}}{\varepsilon} \exp\left( \frac{1}{\varepsilon}\sum_j\varphi_j \right) + O(t_n^2)\right) \exp\left( -\frac{1}{\varepsilon}c \right)\\
    &\to \frac{\sum_{j=1}^J h_j}{\varepsilon}\exp\left( \frac{\sum_j\varphi_j-c}{\varepsilon} \right)
\end{align*}
The Hadamard derivative therefore is 
\[\delta\Phi_\varphi(h) = \frac{\sum_{j=1}^Jh_j}{\varepsilon}\exp\left( \frac{\sum_j\varphi_j-c}{\varepsilon} \right).\]
Therefore, by a first-order Taylor approximation,
\begin{align*}
    &\sqrt{N} \left(\exp\left(\frac{\sum_{j=1}^J\hat{\varphi}^m_{N_j} - c}{\varepsilon}\right) - \exp\left(\frac{\sum_{j=1}^J\varphi^*_j - c}{\varepsilon}\right)\right) \\
    =& \sqrt{N}\thickspace\delta\Phi_\varphi \left(\hat{\varphi}_{N}^m - \varphi^*\right) + o_P\left(\|\hat{\varphi}_N^m - \varphi^*\|_{C^\infty}\right)\\
    =&\delta\Phi_{\varphi^*} \left(\rho_1^{-\frac{1}{2}}\sqrt{N_1}(\hat{\varphi}_{N_1}^m - \varphi^*_1),\ldots,\rho_J^{-\frac{1}{2}}\sqrt{N_J}(\hat{\varphi}_{N_J}^m - \varphi^*_J)\right) + o_P\left(\|\hat{\varphi}_N^m - \varphi^*\|_{C^\infty}\right)\\
    \rightsquigarrow &\delta\Phi_{\varphi^*} \left(\rho^{-\frac{1}{2}}\right)^\top Z_\varphi\\
    =&\exp\left( \frac{\sum_j\varphi^*_j-c}{\varepsilon} \right) \frac{\sum_{j=1}^J\rho_j^{-\frac{1}{2}} Z_{\varphi,j}}{\varepsilon},
\end{align*}
where the weak convergence in line 4 follows from the consistency of the potentials in Lemma \ref{lem:consistency}. The limit expression  is also a mean-zero tight Gaussian process since it is a linear map of $Z_\varphi$.
\end{proof}

\subsection{Proof of Corollary \ref{corr:matching}}
\begin{proof}
We want to derive the asymptotic probability measure of the empirical process
\begin{align*}
    &\sqrt{N}\left( \int f(x) \dd\hat\gamma_N(x) - \int f(x) \dd\gamma(x)  \right)\\
    =&\sqrt{N}\left(\int f(x) \exp\left(\frac{\sum_{j=1}^J\hat\varphi_{N_j}-c}{\varepsilon}\right)\dd\bigotimes_j\hat\mu_{N_j} - \int f(x) \exp\left(\frac{\sum_{j=1}^J\varphi_j^*-c}{\varepsilon}\right)\dd\bigotimes_j\mu_j  \right),
\end{align*}
which is an empirical process indexed by estimated functions \citep[Section 19.4]{wellner2007empirical, van2000asymptotic}.
For notational convenience we use the empirical process notation 
\[\mathbb{P}_N f=\int f(x) \dd\bigotimes_{j=1}^J\hat\mu_{N_j}(x),\qquad\mathbb{P} f=\int f(x) \dd\bigotimes_{j=1}^J\mu_{_j}(x),\qquad\mathbb{G}_Nf = \sqrt{N}(\mathbb{P}_Nf-\mathbb{P}f)\] and write
\[\sqrt{N}\left(\mathbb{P}_N f \hat{K}_N - PfK\right) 
= \mathbb{G}_N\left(f \hat{K}_N - fK\right) + \mathbb{G}_NfK + \sqrt{N}P\left(f \hat{K}_N - fK\right).\]

We proceed term by term. We first show that the first term is asymptotically negligible, i.e.~$\mathbb{G}_N \left(\hat K_Nf -Kf\right) = o_P(1)$. For this we have to verify the conditions in Lemma 19.24 in \cite{van2000asymptotic}.
Similar to the proof of Theorem \ref{thm:potential}, we define the following class of functions
\begin{align*}
    \tilde{\mathcal{F}}_{Kf}\coloneqq\left\{\left[\exp\left(\frac{\sum_{j=1}^J\varphi_j - c}{\varepsilon}\right) f \right]: \thickspace \right. \left.\thickspace \varphi\in \prod_jC^\infty(\mathcal{X}_j), f\in \mathcal{F}\right\},
\end{align*}
where $\mathcal{F}$ is some uniformly bounded Donsker class on $\prod_j\mathcal{X}_j$ by assumption. Using the same argument as in Theorem \ref{thm:potential}, we split this class into 
\[\tilde{\mathcal{K}}_{K}\coloneqq\left\{\left[\exp\left(\frac{\sum_{j=1}^J\varphi_j - c}{\varepsilon}\right) \right]:\thickspace \right. \left.\thickspace \varphi\in \prod_jC^\infty(\mathcal{X}_j)\right\}\qquad\text{and}\qquad \mathcal{F}.\] The first term has a bounded bracketing entropy integral since the exponential function is smooth. Moreover, due to the compactness of $\mathcal{X}_j$ and smoothness of functions in $\tilde{\mathcal{K}}_{\delta,K}$ it is also uniformly bounded and hence a uniformly bounded $\mathbb{P}$-Donsker class. Since the product of two uniformly bounded Donsker classes is a Donsker class by Example 2.10.8 in \cite{wellner2013weak}, this implies that $\tilde{\mathcal{F}}_{\delta,Kf}$ is also $\mathbb{P}$-Donsker. We also need to show that 
\begin{equation}\label{eq:quadconv}P\left(\hat{K}_Nf - Kf\right)^2 = o_P(1).\end{equation}
This follows from the following bound.
\begin{align*}
&\left\lvert \exp\left(\frac{\sum_j\varphi_j-c}{\varepsilon}\right) -\exp\left(\frac{\sum_j\varphi^*_j-c}{\varepsilon} \right)\right\rvert\\
=& \left\lvert \exp\left(\frac{\sum_j\varphi_j^*-c}{\varepsilon}\right)\left[ \exp\left(\frac{\sum_j\varphi_j-\varphi^*_j}{\varepsilon}\right) -1\right]\right\rvert\\
\leq& \exp\left(\frac{\sum_j\left\|\varphi_j^*\right\|_\infty-c}{\varepsilon}\right)\left\lvert \exp\left(\frac{\sum_j\left\|\varphi_j-\varphi^*_j\right\|_\infty}{\varepsilon}\right) -1\right\rvert.
\end{align*}
Taking the supremum over $\prod_j\mathcal{X}_j$ on both sides gives
\begin{align*}
    &\left\| \exp\left(\frac{\sum_j\varphi_j-c}{\varepsilon}\right) -\exp\left(\frac{\sum_j\varphi^*_j-c}{\varepsilon} \right)\right\|_\infty\\\leq& \exp\left(\frac{\sum_j\left\|\varphi_j^*\right\|_\infty-c}{\varepsilon}\right)\left\lvert \exp\left(\frac{\sum_j\left\|\varphi_j-\varphi^*_j\right\|_\infty}{\varepsilon}\right) -1\right\rvert.
    \end{align*}
This implies that the LHS is bounded, because the RHS is by assumption bounded. Therefore, by the boundedness of $f$, H\"older's inequality in conjunction with the consistency of the potentials from Lemma \ref{lem:consistency} implies \eqref{eq:quadconv}. 

Therefore, by Lemma 19.24 in \citet{van2000asymptotic} or Lemma 3.3.5 in \citet{wellner2013weak} the first term is uniformly asymptotically negligible, i.e.
\begin{equation}\label{eq:asyneg}
\sup_{f\in\mathcal{F}}\left\lvert\mathbb{G}_N\left(\hat{K}_Nf-Kf\right)\right\rvert = o_{P^*}\left(1+\sqrt{N}\|\hat{\varphi}^m_N-\varphi^*\|_{C^\infty}\right).
\end{equation}

Now focus on the second and third term jointly. The second term is a standard empirical process over a Donsker class which converges to a tight mean-zero Gaussian process. The third term also converges to a tight Gaussian process by the Continuous Mapping Theorem \citep[Theorem 1.11.1]{wellner2013weak} and Corollary \ref{corr:density}, where we showed weak convergence to a tight mean-zero Gaussian process.

Both weak convergence results are for the marginal distributions, not the joint distribution, however, and in general those two need not be independent. Therefore, we can only conclude that the limit process of the sum of the two terms is a mean-zero tight process.
\end{proof}

\subsection{Proof of Corollary \ref{cor:conditional}}
\begin{proof}
We first prove part A. 
Let $I:=\{i,j\}$, and $I^c:=\{1,..,J\}\setminus I$. 
\begin{align*}
    \dd \gamma(x_{[J]}) &= K(x_1,..,x_J)\dd \bigotimes_{j=1}^J\mu_{j}(x_j)\\
    &=\exp\left\{ \frac{1}{\varepsilon} \left(\sum_{j=1}^J\varphi_j(x_j)-c(x_1,..,x_J)\right)\right\}\dd \bigotimes_{j=1}^J\mu_{j}(x_j)\\
    \dd \gamma(x_I) &= \exp\left\{\frac{1}{\varepsilon}\sum_{j\in I}\varphi_j(x_j)\right\}\int_{\mathcal{X}_{I^c}} \exp\left\{\frac{1}{\varepsilon}\left(\sum_{j\in I^c}\varphi_j(x_j)-c(x_1,..,x_J)\right) \right\}\dd \bigotimes_{j\in I^c}\mu_{j}(x_j) \dd \bigotimes_{j\in I}\mu_{j}(x_j) \\
    &=\exp\left\{\frac{1}{\varepsilon}\sum_{j\in I}\varphi_j(x_j)\right\}
    h(x_I)\dd \bigotimes_{j\in I}\mu_{j}(x_j)
\end{align*}
where $h(x_I):=\int_{\mathcal{X}_{I^c}} \exp\left\{\frac{1}{\varepsilon}\left(\sum_{j\in I^c}\varphi_j(x_j)-c(x_1,..,x_J)\right) \right\}\dd \bigotimes_{j\in I^c}\mu_{j}(x_j)$ 

For any $f\in \mathcal{F}$,
define $g$ on $\mathcal{X}_j$ as 
\[g(x_j):=\int_{\mathcal{X}_i} f(x_i)\exp\left\{\frac{1}{\varepsilon}\sum_{j\in I}\varphi_j(x_j)\right\}
    h(x_I)\dd \mu_{i}(x_i)\]
For any $B\in\mathcal{B}(\mathcal{X}_j)$, using Fubini's theorem,
\begin{align*}
    \mathbb{E}(g(X_j)I_B(X_j)) &= \int_{\mathcal{X}_j} g(x_j)I_B(x_j) \dd\mu_j(x_j) \\
    &= \int_{\mathcal{X}_j} \int_{\mathcal{X}_i} f(x_i) I_B(x_j)\exp\left\{\frac{1}{\varepsilon}\sum_{j\in I}\varphi_j(x_j)\right\}
    h(x_I)\dd \mu_{i}(x_i) \dd\mu_j(x_j)\\
    &= \mathbb{E}(f(X_i)I_B(X_j))
\end{align*}
Hence $g(x_j) = \mathbb{E}(f(X_i)|X_j=x_j)$. Now denote 
\[\hat g_N(x_j) := \int_{\mathcal{X}_i} f(x_i)\exp\left\{\frac{1}{\varepsilon}\sum_{j\in I}\hat\varphi_{N,j}(x_j)\right\}
    \hat h(x_I)\dd \hat\mu_{i}(x_i),\]
to show $\sqrt{N}( \mathbb{E}_{\hat\gamma}(f(X_i)|X_j=x_j)- \mathbb{E}_\gamma(f(X_i)|X_j=x_j)) \rightsquigarrow G_{i|j} $, we only need to show 
\[\sqrt{N} \left( \hat g_N(x_j)-g(x_j)\right) \rightsquigarrow G_{i|j} \]

Expanding $g(x_j)$,
\begin{align*}
    g(x_j)&=\int_{\mathcal{X}_{[-j]}} f(x_i)\exp\left\{\frac{1}{\varepsilon}\left(\sum_{j=1}^J\varphi_j(x_j)-c(x_1,..,x_J)\right) \right\}\dd \bigotimes_{[-j]}\mu_{j}(x_j)\\
&=\int_{\mathcal{X}_{[-j]}} f(x_i) K(x_1,..,x_J)\dd \bigotimes_{[-j]}\mu_{j}(x_j)\\
\hat g_N(x_j) &= \int_{\mathcal{X}_{[-j]}} f(x_i) \hat K_N(x_1,..,x_J)\dd \bigotimes_{[-j]}\hat\mu_{j}(x_j)
\end{align*}
The rest of the proof for part A follows the same way as in Corollary \ref{corr:matching}.

\emph{Proof of Part B.}
Denote $\mathbb{\hat E}[Y(0)|T=1]$ as the estimator for the unobserved outcome under the OT map $\mathbb{E}_\epsilon[Y(0)|T=1]$. Note that here $\mathbb{E}_\epsilon[Y(0)|T=1]$ is not the true counterfactual outcome, but a notation we use to describe the population mean of the estimator $\mathbb{\hat E}[Y(0)|T=1]$. By assumption \ref{ass:causal}, to avoid confusion from the penalty parameter $\varepsilon$, $Y_i(0) = f_0(X_i)+\eta_i$ where the noise term here will be denoted as $\eta_i$. Since the random noise $\eta_i$ is independent from the covariate $X$, we define $\mathbb{\hat E}[Y(0)|T=1]$ and $\mathbb{E}_\varepsilon[Y(0)|T=1]$ in the following way.
\begin{align*}
    \mathbb{\hat E}[Y(0)|T=1] &:= \int_{\mathcal{X}_1} \int_{\mathcal{X}_0}f_0(s)\dd\hat{\gamma}_{\varepsilon,0|1}(s|x) \dd\hat{\mu}_1(x)\\
    \mathbb{E}_\epsilon[Y(0)|T=1] &:= \int_{\mathcal{X}_1} \int_{\mathcal{X}_0}f_0(s)\dd\gamma_{\varepsilon,0|1}(s|x) \dd\mu_1(x)
\end{align*}
Furthermore, based on the conditional density form from Part A,
\begin{align*}
    \mathbb{\hat E}[Y(0)|T=1] 
    &=\int_{\mathcal{X}_0\times \mathcal{X}_1} f_0(s)\hat{K}(s,x)\dd\hat{\mu}_0(s)\dd\hat{\mu}_1(x)
\end{align*}
Part B now follows directly from Corollary \ref{corr:matching}.
\vspace{-0.25cm}\end{proof}







\end{document}